\documentclass[11pt, titlepage]{article}
\usepackage{eurosym}
\usepackage{amsmath,amssymb,amsthm,amsfonts}
\usepackage{subcaption}
\usepackage{appendix}
\usepackage{amsmath}
\usepackage{mathrsfs}
\usepackage{amsfonts}
\usepackage{amssymb}
\usepackage{multirow}
\usepackage{rotating}
\usepackage{xcolor}
\usepackage{lscape}
\usepackage{verbatim}
\usepackage{graphicx}
\usepackage{cancel}
\usepackage[normalem]{ulem}
\usepackage{url}
\usepackage{booktabs}
\usepackage{algorithm}
\usepackage{algpseudocode}

\numberwithin{equation}{section} 

\newtheorem{theorem}{Theorem}[section]
\newtheorem{corollary}{Corollary}[section]
\newtheorem{lemma}{Lemma}[section]
\newtheorem{proposition}{Proposition}[section]
\newtheorem{definition}{Definition}[section]
\newtheorem{remark}{Remark}
\newtheorem{example}{Example}
 
\usepackage{tikz}

\newcommand{\squared}[1]{\tikz[baseline=(char.base)]{\node[draw,rectangle,inner sep=2pt] (char) {#1};}}

\def\XX{\boldsymbol{X}}

\def\xx{\boldsymbol{x}}
\def\EE{\boldsymbol{E}}

\def\rr{\boldsymbol{r}}
\def\ee{\boldsymbol{e}}
\def\FF{\boldsymbol{F}}

\def \UU{\boldsymbol{U}}
\def\RRR{\boldsymbol{R}}

\def\WW{\boldsymbol{W}}

\def\uu{\boldsymbol{u}}
\def\ii{\boldsymbol{i}}

\def\II{\boldsymbol{I}}

\def\FFF{\mathcal{F}}
\def\GGG{\mathcal{G}}
\def\EEE{\mathcal{E}_d}

\def\yy{\boldsymbol{y}}

\def\pp{\boldsymbol{p}}

\def\RR{\mathbb{R}}

\def\BBB{\mathcal{B}}
\def\ZZ{\boldsymbol{Z}}

\def\ff{\boldsymbol{f}}

\def\design{\mathcal{X}_d}

\def\VaR{\text{VaR}}
\def\ES{\text{ES}}

\def\FGM{\text{FGM}}
\def\GFGM{\text{GFGM}}

\begin{document}

	\title{Generalized FGM dependence: Geometrical representation  and convex bounds on sums}
	
	\author{Hélène Cossette \\
			\textit{\'Ecole d'actuariat, Université Laval, Québec, Canada},\\
			\texttt{helene.cossette@act.ulaval.ca} \and
			Etienne Marceau \\
			\textit{\'Ecole d'actuariat, Université Laval, Québec, Canada},\\
			\texttt{etienne.marceau@act.ulaval.ca} \and
			Alessandro Mutti  \\ 
			\textit{Department of Mathematical Sciences G. Lagrange, Politecnico di Torino, Italy}, \\
			\texttt{alessandro.mutti@polito.it} \and
	        Patrizia Semeraro \\ 
			\textit{Department of Mathematical Sciences G. Lagrange, Politecnico di Torino, Italy},\\
			\texttt{patrizia.semeraro@polito.it}
			}
	\maketitle

	\begin{abstract}
		
		Building on the one-to-one relationship between generalized FGM copulas and multivariate Bernoulli distributions, we prove that the class of multivariate distributions with generalized FGM copulas is a convex polytope.
		Therefore, we find sharp bounds in this class for many aggregate risk measures, such as value-at-risk, expected shortfall, and entropic risk measure, by enumerating their values on the extremal points of the convex polytope.
		This is infeasible in high dimensions.
		We overcome this limitation by considering the aggregation of identically distributed risks with generalized FGM copula specified by a common parameter $p$.
		In this case, the analogy with the geometrical structure of the class of Bernoulli distribution allows us to provide sharp analytical bounds for convex risk measures.
		
	\end{abstract}

	\section{Introduction}

	Finding bounds for aggregated risks with partial information on their joint distribution is a widely addressed problem in insurance and finance. 
	The available information on the multivariate dependence is often modeled using a copula. Our main contribution is to solve the problem of finding analytical bounds for aggregated risks under the assumption that their dependence is modeled using a generalized Farlie-Gumbel-Morgenstern (GFGM) copula. 
	
	
	FGM copulas allow both positive and negative dependence, but do not cover the complete range of dependence. Due to its simple analytical formulation, the FGM family of copulas is used to build multivariate models, for example, in insurance and finance \cite{hashorva1999extreme}, \cite{mao2015risk}, \cite{yang2013extremes}, transport \cite{zou2014constructing}, natural hazard - drought \cite{saghafian2014drought}, marketing \cite{kim2022copula}.
	
	There exists a wide variety of extensions of bivariate FGM copulas, such as the well-known Huang-Kotz FGM copulas, but much less in higher dimension. 
	Recently, researchers have integrated the Huang-Kotz FGM copulas and other extensions of that family as building blocks within models for directional dependence of genes in bioinformatics (\cite{kim2008copula}), bivariate failure time models with competing risks (\cite{shih2018likelihood}, \cite{shih2019bivariate}, \cite{shih2019package}), 
	copula-based stress-strength models in reliability (\cite{domma2013copula}, \cite{hudaverdi2023copula}), and models that account for asymmetric dependence
	between variables in biomedicine (\cite{li2023property}). 
	Through a stochastic representation involving multivariate Bernoulli vectors, the authors of \cite{blier2024new} introduce a family of multivariate GFGM copulas. Such a stochastic representation is key to grasp the dependencies between the components of a random vector and to facilitate greatly the sampling procedure. It also implies that 
	the class of GFGM copulas inherits the geometrical properties of the class $\BBB_d(\pp)$ of probability mass functions (pmfs) of $d$ dimensional Bernoulli vectors with mean vector $\pp=(p_1,\ldots, p_d)$. Moreover, the vector $\pp$ becomes a vector parameter for $\GFGM$ copulas. Although some results are presented for the general case, we focus on the case $p=p_1=\ldots=p_d$, for two reasons: some results cannot be generalized in a straightforward manner to have different values within $\pp$; and we aim at keeping the model parsimonious in terms of number of parameters.
	This last motivation allows us to have a better insight in the role of the parameter $p$ driving the dependence. Actually, we do more. In fact we also prove that the class $\mathcal{G}_d^p(F)$ of joint distributions with a common univariate margin $F$ and GFGM$(p)$ copula shares the same geometrical structure as the class of multivariate Bernoulli distributions.
	This analogy allows us to easily work with sums $S = X_1+\dots+X_d$ of random variables with joint distribution in $\mathcal{G}_d^p(F)$, applying the results of \cite{fontana2021model}.
	We also show that the convex order is preserved from the elements of the class of sums of components of random vectors following Bernoulli distributions to our class of interest $\mathcal{S}_d^{p}(F)$, that is the class of distributions of sums of random vectors with cumulative distribution function (cdf) in $\mathcal{G}_d^p(F)$.
	This contribution is our main result and allow us to find the vectors in $\mathcal{G}_d^p(F)$ whose sums are minimal in convex order and this is important because they are the vectors where the lower bounds for convex risk measures are reached. 
	We considered two convex risk measures, the expected shortfall, and the entropic risk measure, and also the value-at-risk. We analytically find their bounds in the cases of exponential margins and discrete margins and provide numerical illustrations in these two special cases in high dimensions. Last but not least, building on the geometrical structure of the joint distributions behind the sums, we address another important open issue and exhibit some possible alternative dependence structures corresponding to minimal aggregated risks.
	
	We also present some simple but new results for the GFGM dependence without the assumption of a scalar parameter $p$ and without assuming identically distributed risks.
	We find a stochastic representation that generalizes the correspondence proved in \cite{blier2024new} to any random vector $\XX=(X_1,\ldots, X_d)$ with dependence structure defined by a GFGM copula. 
	Building on this representation, we also prove that the class $\mathcal{G}^{\pp}_d(F_1,\ldots,F_d)$ of joint distributions of random vectors $\XX$  with univariate marginals $F_1,\dots, F_d$ and GFGM copula with parameters $\pp=(p_1,\dots, p_d)$ is a convex polytope, that is a convex hull of a finite set of points, called extremal points. 
	However, the number of extremal points is huge and there are computational limitations in finding them in high dimensions.
	Overcoming these theoretical limitations is left for future work.
	Although the theoretical investigation of this case is beyond the scope of this paper, we also discuss some numerical examples of the GFGM dependence, without the assumption of identically distributed risk. In this general case, we can proceed by enumeration of the extremal points up to dimension $d=5$.
	The paper is structured as follows. Section \ref{prelim} introduces the preliminary notions about multivariate Bernoulli distributions and GFGM copulas with a common parameter and  their link.
	A new stochastic representation for GFMG copulas is provided in Section \ref{newrep}.
	We study the geometrical structure and the convex order in the class of uniform vectors with GFGM copulas with a common parameter $p$ in Section \ref{FGMcop}. In the subsequent section, we provide sharp bounds for the convex risk measures and the value-at-risk and we provide numerical illustrations. The last Section \ref{NumILL} presents an example of the generalization to different values of $p$ for future research purposes and concludes.

	\section{Preliminaries and state of the art}\label{prelim}
	In this section, we recall some notions on the set $\BBB_d$  of $d$-dimensional pmfs which have Bernoulli univariate marginal distributions and on the class $\mathcal{C}_d$ of GFGM copulas.

	\subsection{Multivariate Bernoulli distributions and convex polytopes}
	
	Let us consider the Fr\'echet class  $\BBB_d(\pp)= \BBB_d(p_1,\ldots, p_d)$  of multivariate Bernoulli distributions with  Bernoulli marginal distributions with means $p_j, j \in \{1,\ldots, d\}$.
	We assume throughout the paper that $p_j$ are rational, that is $p_j\in \mathbb{Q}$, $j \in \{1,\ldots, d\}$. Since $\mathbb{Q}$ is dense in $\RR$, this is not a limitation in applications. 
	We denote by $\BBB_d(p)$ the class of multivariate Bernoulli distributions with identical Bernoulli marginal distributions with mean $p$, meaning $p_1 =\ldots=p_d=p$.
	We assume that vectors are column vectors and we denote by $A^{\top}$ the transpose of a matrix $A$.

	If  $\II=(I_1, \dots, I_d)$ is  a random vector with joint pmf $f$ in $\BBB_d$, we denote
	the column vector which contains the values of $f$ over $\design=\{0, 1\}^d$  by $\ff =(f_{\xx}:\xx\in\design):=(f(\xx):\xx\in\design)$; we make the non-restrictive assumption that the set $\design$ of $2^d$ binary vectors is ordered according to the reverse-lexicographical criterion.  
	As an example, we consider  $d=3$ and we have $\mathcal{X}_3=\{000, 100, 010, 110, 001, 101, 011, 111\}$.
	The notations $\II \in \BBB_d(\pp)$  and  $\ff \in \BBB_d(\pp)$ indicate that $\II$  has joint pmf $f\in \BBB_d(\pp)$.
	In \cite{fontana2018representation}, the authors prove that $\BBB_d(\pp)$ is a convex polytope (see as a standard reference \cite{de1997computational}); it means that $\BBB_d(\pp)$ admits the following representation:
	\begin{equation*}  \label{cone}
		\BBB_d(\pp)=\{\ff\in \RR_+^{2^d}: H \ff=0, \sum_{\xx\in \design} f_{\xx}=1\},
	\end{equation*}
	where $H$ is a $d \times 2^d$ matrix whose rows, up to a non-influential multiplicative constant, are $((\boldsymbol{1}-\xx_j)^{\top} - \tfrac{(1-p_j)}{p_j}\xx_j^{\top})$, $j\in\{1,\ldots,d\}$, and where $\xx_j$ is the vector which contains only the $j$-th element of $\xx \in \design$, $j\in\{1,\ldots,d\}$, e.g. for the bivariate case $\xx_1^T=(0, 1,0,1)$ and $\xx_2^T=(0, 0,1,1)$.
	Therefore, there are joint pmfs $\rr_k \in \BBB_d(\pp)$, $k \in \{1,\ldots, n_{\pp}^{\BBB}\}$, and for any $\ff\in \BBB_d(\pp)$, there exist $\lambda_1,\ldots, \lambda_{n_{\pp}^{\BBB}}\geq0$ summing up to one such that
	\begin{equation}
		\ff=\sum_{k=1}^{n_{\pp}^{\BBB}}\lambda_k \rr_k.
		\label{eq:ConvexRepresentationRays}
	\end{equation}
	We call the vectors $\rr_k$, $k \in \{1,\ldots, n_{\pp}^{\BBB}\}$, the extremal points of $\BBB_d(\pp)$, and $r_k$ the corresponding joint pmfs of the random vector $\RRR_k$. Here, $n_{\pp}^{\BBB}$ is the number of extremal points of $\BBB_d(\pp)$ which depends on $\pp$ and obviously on $d$.
	For low dimension $d$, the extremal points $\rr_k$ can be found using the software \texttt{4ti2} (see \cite{4ti2}). 
	However, their number $n_{\pp}^{\BBB}$ increases rapidly with the dimension $d$, as discussed in Section 2 of \cite{fontana2024high}.  
	
	We need to introduce the following classes of distributions as they are building blocks of one of our main results. Let $\mathcal{E}_d(p) \subseteq \BBB_d(p)$ be the class of exchangeable $d$-dimensional Bernoulli distributions with mean $p \in [0,1]$. Note that $f\in \mathcal{E}_d(p)$ if $f\in\BBB_d(p)$ and $f(\xx)=f(\xx_{\sigma})$, where $\xx_{\sigma}$ is any permutation of $\xx$, for every $\xx\in \design$. We also denote by $\mathcal{D}_d(dp)$ the class of univariate discrete distributions with support on $\{0,1,\ldots,d\}$ and mean $dp$. 
	If $D$ is a discrete random variable with pmf $f_D$ in $\mathcal{D}_d(dp)$, we denote the column vector containing the values of its pmf $f_D$ over $\{0,\ldots, d\}$ by $\ff^D = (f_1^D,\ldots, f_{d+1}^D)^{\top}$. 
	In other words, $f_D(k) = f_{k+1}^D$, $k \in \{0,1,\dots,d\}$.
	The notations $D\in \mathcal{D}_d(dp)$ and $\ff^D \in\mathcal{D}_d(dp)$ indicate that the discrete random variable $D$ has pmf $f_{D}\in \mathcal{D}_d(dp)$.
	The authors of \cite{fontana2021model} prove that $\mathcal{D}_d(dp)$ is a convex polytope: 
	it means that $\ff^{D}\in \mathcal{D}_d(dp)$ if and only if there
	exist $\lambda_1, \ldots, \lambda_{n_p^{\mathcal{D}}}\geq 0$ summing up to 1 such that
	\begin{equation*}\label{Sgenerators}
		{\ff}^{D}=\sum_{k=1}^{n_p^{\mathcal{D}}} \lambda _k \rr^{D}_k,
	\end{equation*}
	where $\rr^{D}_k, \, k \in \{1,\ldots, n_p^{\mathcal{D}} \}$, are the extremal points of $\mathcal{D}_d(dp)$ and $n_p^{\mathcal{D}}$ is their number which depends on $p$ and obviously on $d$ (see Corollary 4.6 in \cite{fontana2021model} for the computation of $n_p^{\mathcal{D}}$). 
	We denote by ${R}^{D}_k$ a random variable with pmf  $r_{D,k}$, $k \in \{1,\dots,n_p^{\mathcal{D}}\}$. 
	The extremal points have at most two non-zero components.
	Let $k_1$ and  $k_2$ with $k_1=0,1,\ldots, k_1^{\vee}$, $k_2=k_2^{\wedge}, k_2^{\wedge}+1, \ldots, d$, where  $k_1^{\vee}$ is
	the largest integer lower than $dp$ and $k_2^{\wedge}$ is the smallest integer
	greater than $dp$. The extremal pmfs $r_{D, k}$ have support on $\{k_1, k_2\}$ and  have the following analytical expression:

	\begin{equation*}  \label{binul}
		r_{D, {k}}(y)=
		\begin{cases}
			\frac{k_2-dp}{k_2-k_1}, & y=k_1 \\
			\frac{dp-k_1}{k_2-k_1}, & y=k_2 \\
			0, & \text{otherwise}%
		\end{cases}.
	\end{equation*}
	
	If $dp$ is an integer, we also have an extremal point with support on the point $pd$, that is
	\begin{equation*}  \label{onenul}
		r_{D, dp}(y)=
		\begin{cases}
			1, & y=dp \\
			0, & \text{otherwise}%
		\end{cases}
		.
	\end{equation*}
	
	In \cite{fontana2021model}, the authors show that the following relationship between classes of distributions holds:
	\begin{equation}\label{iff}
		\mathcal{E}_d(p) \leftrightarrow \mathcal{D}_d(dp),
	\end{equation}
	that is, the class $\mathcal{D}_d(dp)$ has a one-to-one relationship with the class of exchangeable Bernoulli distributions $\mathcal{E}_d(p)$. 
	In \cite{fontana2021model}, the authors also show that given $D\in \mathcal{D}_d(dp)$ there is one and only one exchangeable element $\II^e\in \BBB_d(p)$ such that $\sum_{j=1}^d I^e_j\overset{\mathcal{L}}{=} D$, where the notation $\overset{\mathcal{L}}{=}$ indicates equality in distribution. 
	It follows that $\mathcal{E}_d(p)$ is a convex polytope and that the extremal points of $\mathcal{E}_d(p)$ are the exchangeable pmfs corresponding to the extremal points of  $\mathcal{D}_d(dp)$. We denote by $\ee_k$ extremal points or extremal pmfs of $\EEE(p)$, $k \in \{1,\ldots,n_p^{\mathcal{D}}\}$. 
	We denote by $\boldsymbol{E}_k$ a random vector with pmf  $e_k$. 
	In Table \ref{extremPointgen}, we provide each convex polytope with its generators.
	
	\begin{table}[htb]
		\centering
		\begin{tabular}{c|c|c|c}
			\toprule
			Polytope & pmf-generator & rv-generator & number of generators\\
			\midrule
			$\BBB_d(p)$ & $\rr_k\in \RR^{2^d}$ & $\RRR_k$ & $n_p^{\BBB}$ \\
			$\mathcal{D}_d(dp)$& $\rr_{D,k}\in\RR^{d+1}$  & $R_{D,k}$ & $n_p^{\mathcal{D}}$\\
			$\EEE(p)$ & $\ee_k \in\RR^{2^d}$ & $\EE_k$ & $n_p^{\mathcal{D}}$ \\
			\bottomrule
		\end{tabular}
		\caption{\emph{For each polytope of pmfs the second column provides the name of extremal points and their dimension, the third column provides the name of its corresponding random variable, and the last column the number of generators. }}
		\label{extremPointgen}
	\end{table}

	\subsection{Bernoulli distributions and GFGM copulas} \label{GFGMprel}

	The authors of \cite{blier2024new} provide a stochastic representation for GFGM copulas building on a multivariate Bernoulli vector $\II\in \mathcal{B}_d(p)$ establishing a link with the class $\BBB_d(p)$ on which we establish our results.
	Let $\II \in \BBB_d(\pp)$ be a $d$-variate Bernoulli random vector and let $\UU_0 = (U_{0, 1}, \dots, U_{0, d})$ and $\UU_1 = (U_{1, 1}, \dots, U_{1, d})$ be vectors of $d$ independent uniform random variables.
	Assume the random vectors $\II$, $\UU_0$, and $\UU_1$ to be independent and define the random vector $\UU$ with the following representation:
	\begin{equation}\label{eq:representation-u}
		\UU \overset{\mathcal{L}}{=} \UU_0 ^{{1} - {\pp}} \UU_1^{\II} = (U_{0, 1}^{1 - p_1}U_{1, 1}^{I_1}, \dots, U_{0, d}^{1 - p_d}U_{1, d}^{I_d}).
	\end{equation}
	The joint cdf of the random vector $\UU$ in \eqref{eq:representation-u} is the GFGM copula $C$ with vector of parameters $\pp$ as defined below.
	\begin{definition}\label{eq:copula-GFGM}
		A $d$-variate GFGM copula $C$ with vector of parameters $\pp=(p_1,\ldots,p_d) \in (0,1)^d$ has the following expression:
		\begin{equation} \label{eq:copula-natural}
			C(\uu) = 
			\prod_{m = 1}^d u_m 
			\left( 1 + \sum_{k = 2}^d \sum_{1 \leq j_1 < \dots < j_k \leq d} \nu_{j_1\dots j_k} \left(1 - u_{j_1}^{\frac{p_{j_1}}{1-p_{j_1}}}\right) \cdots \left(1 - u_{j_k}^{\frac{p_{j_k}}{1-p_{j_k}}}\right) \right),
		\end{equation}
		for $\uu \in [0, 1]^d,$ where, for $1\leq j_1 < \dots < j_k \leq d, k \in \{2, \dots, d\}$,
		\begin{equation*}
			\nu_{j_1\dots j_k} = \mathrm{E}\left[\prod_{n = 1}^{k} \frac{I_{j_n} - p_{j_n}}{p_{j_n}}\right],
		\end{equation*}
		where $\II = (I_1,\dots,I_d) \in \BBB_d(\pp)$.
	\end{definition}
	
	If one lets $b_j = \frac{p_j}{1-p_j}$, for $j \in \{1,\dots,d\}$, the expression of the copula $C$ in \eqref{eq:copula-natural} becomes
	\begin{equation} \label{eq:copula-natural2}
		C(\uu) = 
		\prod_{m = 1}^d u_m 
		\left( 1 + \sum_{k = 2}^d \sum_{1 \leq j_1 < \dots < j_k \leq d} \nu_{j_1\dots j_k} \left(1 - u_{j_1}^{b_{j_1}}\right) \cdots \left(1 - u_{j_k}^{b_{j_k}}\right) \right),
	\end{equation}
	for $\uu \in [0, 1]^d$.
	The shape parameters $p_j$ (or $b_j$), $j \in \{1,2,\dots,m\}$, govern the dependence relation between the components of $\boldsymbol{U}$ in regard to the symmetry or asymmetry.
	When $p_j = p$, implying that $b_j = b = \frac{p}{1-p}$, $j \in \{1,2,\dots,m\}$, the copula $C$ in \eqref{eq:copula-natural2} corresponds to the multivariate version of the symmetric bivariate Huang-Kotz FGM copula introduced and studied in Section 2 of \cite{huang1999modifications}. In \cite{bairamov2001new}, the authors propose an asymmetric bivariate Huang-Kotz FGM copula while \cite{bekrizadeh2012new} suggest a multivariate version of the symmetric Huang-Kotz FGM copula; without however providing lower and upper bounds on the set of dependence parameters. The stochastic representation in \eqref{eq:representation-u} allows the authors of \cite{blier2024new} to provide a multivariate extension of the Huang-Kotz FGM family of copulas with constraints on the dependence parameters; it is a key result to investigate dependence properties within this copula family and it easily provides a sampling algorithm. In the present paper, we build on the stochastic representation of the \eqref{eq:representation-u} to provide results on bounds of dependent aggregated risks. For a review on various extensions of the family of FGM copulas including the Huang-Kotz FGM copulas, see \cite{saminger2021impact}  and \cite{blier2024new}.


	We denote by $\mathcal{C}_d^{\pp}$ and $\mathcal{C}_d^p$ the classes of $d$-variate GFGM copulas with parameters $\pp = (p_1,\ldots,p_d)$ and with a common parameter $p$, respectively.
	The notation $\UU \in \mathcal{C}_d^{\pp}$ indicates that the joint cdf $C$ of the random vector $\UU$ with uniform margins is a copula $C \in \mathcal{C}_d^{\pp}$.
	In the special case  $p=\tfrac{1}{2}$, $\mathcal{C}_d^{1/2} \equiv \mathcal{C}_d^{\FGM}$, where $\mathcal{C}_d^{\FGM}$ is the class of $d$-variate Farlie-Gumbel-Morgenstern (FGM) copulas, see e.g. Section 6.3 of \cite{durante2015principles} and \cite{blier2022stochastic}.

	In \cite{blier2024new}, the authors mention in Remark 1 that the class $\mathcal{C}^{\pp}_d$ shares the geometrical structure of $\mathcal{B}_d(\pp)$. In particular, any $F_{\UU}\in \mathcal{C}_d^{\pp}$ is a convex combination of  $F_{\UU^{(\RRR_k)}}$, where $\UU^{(\RRR_k)}$ is built from $\RRR_k$ according to \eqref{eq:representation-u}, $k \in \{1,\ldots,n_{\pp}^{\BBB}\}$.  
	Indeed, by using the stochastic representation in \eqref{eq:representation-u}, there is a multivariate Bernoulli vector $\II\in \BBB_d(\pp)$ such that $\UU = \UU^{(\II)} = \UU_0 ^{{1} - {\pp}} \UU_1^{\II}$, and we have
	\begin{equation*}\label{eq:politCop}
		\begin{split}
			F_{\UU^{(\II)}}(\xx) &= \Pr \bigg( U_1^{(\II)} \leq x_1,\ldots, U_d^{(\II)} \leq x_d \bigg) \\
			& \overset{*}{=}\sum_{\ii \in \{0,1\}^d} \Pr \bigg( U_{0,1}^{1-p_1}U_{1,1}^{i_1} \leq x_1,\ldots, U_{0,d}^{1-p_d}U_{1,d}^{i_d}\leq x_d \bigg)f_{\II} ( \ii)  \\
			& = \sum_{k=1}^{n_p^{\mathcal{B}}}\lambda_k \sum_{\ii \in \{0,1\}^d} \Pr  \bigg( U_{0,1}^{1-p_1} U_{1,1}^{i_1} \leq x_1,\ldots, U_{0,d}^{1-p_d}U_{1,d}^{i_d}\leq x_d \bigg) r_k (\ii) \\
			& \overset{**}{=} \sum_{k=1}^{n_p^{\mathcal{B}}}\lambda_k F_{\UU^{(\RRR_k)}}(\xx), \qquad \xx\in [0,1]^d,
		\end{split}
	\end{equation*}
	where equalities  $\overset{*}{=}$ and  $\overset{**}{=}$ respectively follow from the independence of $\UU_0$, $\UU_1$ and $\II$ and of $\UU_0$, $\UU_1$ and $\RRR_k$, $k \in \{1,\ldots, n_{\pp}^{\BBB}\}$. 
	It follows that any GFGM copula  $C$ admits the representation
	\begin{equation}\label{eq:CopConvex}
		C(\uu)
		=
		\sum_{k=1}^{n_{\pp}^{\BBB}}
		\lambda_k C_{\RRR_k}(\uu),
		\quad
		\uu \in [0,1]^d,
	\end{equation}
	where $C_{\RRR_k}$ is the copula associated to $\UU^{(\RRR_k)}$, $k \in \{1,\ldots, n_{\pp}^{\BBB}\}$.
	The class $\mathcal{C}_d^{\pp}$ of copulas is a convex polytope.
	
	We conclude this section with a simple but general new result for a class of multivariate distributions with dependence structure built with a family of copulas being a convex polytope. Our aim is to later investigate the specific class $\GGG_d^{\pp}(F_1,\dots,F_d)$ of distributions with marginal distributions $F_1, \ldots, F_d$ and with a copula in the class $\mathcal{C}_d^{\pp}$.
	
	\begin{proposition}
		Let $\mathcal{C}$ be a class of copulas. Let $\FFF_d(F_1,\dots,F_d)$ be a class of multivariate distributions with marginal distributions $F{_1}, \ldots, F{_d}$ and with a copula in the class $\mathcal{C}$.
		If $\mathcal{C}$ is a convex polytope with extremal points $\widetilde{C}_{1} ,...,\widetilde{C}_{n}$, then $\FFF_d(F_1,\dots,F_d)$ is a convex polytope with extremal points $\widetilde{F}_{1},...,\widetilde{F}_{n}$, where
		\begin{equation*}
			\widetilde{F}_{i}(\xx)=\widetilde{C}_{i}(F_1(x_1),...,F_d(x_d))
		\end{equation*}
		for $\xx \in \RR^d$ and $i \in \{1,\ldots,n\}$. 
	\end{proposition}

	\begin{proof}
		Consider $F_{\XX} \in \FFF_d(F_1,\dots,F_d)$. Then, for $\xx \in \RR^d$, we have
		\begin{align*}
			F_{\XX}(\xx)&=C(F_1(x_1),...,F_d(x_d)) \\
			&=\sum_{i=1}^{n} \lambda_i \widetilde{C}_{i}(F_1(x_1),...,F_d(x_d)),
		\end{align*}
		for some $\lambda_1,...,\lambda_n \geq 0$ such that $\sum_{i=1}^{n} \lambda_i=1$. Define $ \widetilde{F}_{i}(\xx)=\widetilde{C}_{i}(F_1(x_1),...,F_d(x_d))$ and the desired result directly follows.
	\end{proof}

	\begin{corollary}\label{corollary3.1}
		The class $\GGG_d^{\pp}(F_1,\dots,F_d)$ of distributions is a convex polytope with extremal points the distributions associated to the extremal points $\rr_k$ of $\BBB_d(\pp)$, $k \in \{1,\ldots, n^{\mathcal{B}}_{\pp}\}$.
	\end{corollary}
	
	\section{A new representation theorem and consequences}\label{newrep}
	
	We generalize the results of Section \ref{GFGMprel} and introduce a stochastic representation for any random vector $\XX$ with distribution in $\GGG_d^{\pp}(F_1,\dots,F_d)$. We then use this representation in two particular cases, more precisely with exponential and discrete marginals for which the expression of the distribution of the sum is analytical.
	The notation $\XX \in \GGG_d^{\pp}(F_1,\dots,F_d)$ indicates that the cdf of $\XX$, $F_{\XX}$, belongs to the class $\GGG_d^{\pp}(F_1,\dots,F_d)$.
	
	\begin{theorem} \label{thm_stoch_repre_X}
		Let $V_{0,1},\dots,V_{0,d},V_{1,1},\dots,V_{1,d}$ be independent random variables, where $V_{0,j} \sim Beta\left(\frac{1}{1-p_j},1\right)$, $p_j \in (0,1)$, and $V_{1,j}\sim U(0,1)$, for $j \in \{1,\dots,d\}$. 
		Fix some margins $F_1,\ldots,F_d$ and, for $h \in \{0,1\}$, let $\ZZ_h=(Z_{h,1},\ldots, Z_{h,d})$ be vectors of independent random variables with 
		$Z_{0,j} \overset{\mathcal{L}}{=} F^{-1}_j(V_{0,j})$ and $Z_{1,j} \overset{\mathcal{L}}{=} F^{-1}_j(V_{0,j}V_{1,j})$, for all $j \in \{1,\ldots,d\}$.
		Define the random vector
		\begin{equation}\label{stoch_repre_X}
			\XX=(\boldsymbol{1}-\II)\ZZ_0+\II \ZZ_1,
		\end{equation}
		where $\II\in \BBB_d(\pp)$, $\pp = (p_1,\dots,p_d)$. 
		Then we have $F_{\XX} \in \GGG_d^{\pp}(F_1,\ldots,F_d)$, that is the distribution of the random vector $\XX$ has margins $F_1,\dots, F_d$ and copula $C\in \mathcal{C}_d^{\pp}$.
	\end{theorem}
	\begin{proof}
		Let $\UU\in \mathcal{C}_d^{\pp}$ and hence $\UU=\UU_0^{1-\pp}\UU^{\II}$, where $\II\in \BBB_d(\pp)$, with $\pp = (p_1,\dots,p_d)$. In the proof of Theorem 2 of \cite{blier2024new}, we have
		\begin{equation*}
			F_{U_j}(u)=\Pr(U_{0,j}^{1-p_j}U_{1,j}^{I_j}\leq u)=\Pr(I_j=0)F_{U_j|I_j=0}(u)+\Pr(I_j=1)F_{U_j|I_j=1}(u),
		\end{equation*}
		where 
		\begin{equation*}
			F_{U_j|I_j=0}(u)=u^{\frac{1}{1-p_j}}
		\end{equation*}
		and 
		\begin{equation*}
			F_{U_j|I_j=1}(u)=\frac{u}{p_j}-\frac{1-p_j}{p_j}u^{\frac{1}{1-p_j}}, \quad j \in \{1,\ldots, d\}.
		\end{equation*}
		Thus $(U_j|I_j=0)\overset{\mathcal{L}}{=}V_{0,j}$ and $(U_j|I_j=1)\overset{\mathcal{L}}{=}V_{0,j}V_{1,j}$, $j \in \{1,\ldots, d\}$, and the assert follows.
	\end{proof} 
	
	Obviously, the special case of uniform margins leads to a stochastic representation of the GFGM family of copulas equivalent to \eqref{eq:representation-u}.
	We notice that for $p=\tfrac{1}{2}$ we find the stochastic representation of FGM copulas in \cite{blier2024exchangeable}.
	
	Theorem \ref{thm_stoch_repre_X} provides a very useful representation of the random vector $\XX$. This helps us to derive plenty of results such as examining the distribution of any integrable function of $\XX$ and analyzing pairwise dependence properties of $\XX$. The following corollary considers the expectation of functionals of $\XX$ for which we derive bounds, in Section \ref{sec:SharpBounds}, building on the geometrical structure of Bernoulli vectors, the distribution of the sum $S^{(\XX)} = X_1 + \dots + X_d$ which may represent aggregate risks in a portfolio and some results on correlation between components of $\XX$ to analyze dependence corresponding to minimal aggregate risks. 
	
	\begin{corollary}
		Let $\XX\in \mathcal{G}^{\pp}_d(F_1, \ldots, F_d)$, and let $r_k$, $k \in \{ 1,\dots,n_{\pp}^{\BBB} \}$, be the extremal pmfs of $\BBB_d(\pp)$. The following holds.
		\begin{enumerate}
			\item \label{PointEphi} Let $\varphi \colon \RR^d \to \RR$ be a real-valued function for which the expectation exists, we have
			\begin{align}\label{eq:ephi}
				E[\varphi(X_1,\dots,X_d) ] 
				& = \sum_{k=1}^{n_{\pp}^{\BBB}} 
				\lambda_k 
				\sum_{\boldsymbol{i} \in \{0,1\}^d} 
				r_k( \ii)
				E[\varphi(Z_{i_1,1},\dots,Z_{i_d,d})],
			\end{align}
			
			\item \label{PointFS} The distribution of the sum $S^{(\XX)}=\sum_{j}^dX_j$ is given by
			\begin{align}
				F_{S^{(\XX)}}(y) 
				=
				\sum_{k=1}^{n_{\pp}^{\mathcal{B}}}\lambda_k 
				\sum_{\ii \in \{0,1\}^d} r_k(\ii) F_{\sum_{j=1}^d Z_{i_j,j}} (y),
				\quad 
				y \in \RR.
				\label{eq:FSwithExtremalBeautiful}
			\end{align}
			\item \label{Point_Corr} The covariance  between each pair of components $(X_{j_1},X_{j_2})$ of $\XX$ is
			\begin{equation}\label{covariance}
				Cov(X_{j_1},X_{j_2}) 
				= Cov(I_{j_1},I_{j_2}) \gamma_{j_1,j_2}
				= \nu_{j_1,j_2} p_{j_1} p_{j_2} \gamma_{j_1,j_2},
			\end{equation}
			where
			\begin{equation*}
				\gamma_{j_1,j_2} = (E[Z_{1,j_1}] -E[Z_{0,j_1}])
				(E[Z_{1,j_2}] -E[Z_{0,j_2}]),
				\quad 1 \leq j_1 < j_2 \leq d.
			\end{equation*}
			The sharp bounds of the covariance are
			\begin{equation*} \label{eq:covariance_bounds}
				(\max(p_{j_1}+p_{j_2}-1,0)-p_{j_1}p_{j_2} ) \gamma_{j_1,j_2} \leq
				Cov(X_{j_1},X_{j_2}) \leq 
				p_{j_1}(1-p_{j_2}) \gamma_{j_1,j_2},
				\quad 1 \leq j_1 < j_2 \leq d.
			\end{equation*}
			
			\item If $\XX$ has continuous marginal distributions, Spearman's rho between each pair of components $(X_{j_1},X_{j_2})$ of $\XX$, $1 \leq j_1 < j_2 \leq d$, is
			\begin{align*} \label{eq:SpearmanRho}
				\rho_S(X_{j_1},X_{j_2}) 
				&=
				\frac{3 \, Cov(I_{j_1},I_{j_2})}{(2-p_{j_1})(2-p_{j_2})}
				= \frac{3 \, \nu_{j_1,j_2} p_{j_1} p_{j_2}}{(2-p_{j_1})(2-p_{j_2})}
				,
			\end{align*}
			where $\II = (I_1,\dots,I_d)$ is the Bernoulli random vector corresponding to $\XX$ of Theorem~\ref{thm_stoch_repre_X}. The sharp bounds of the Spearman's rho are:
			\begin{equation*} \label{eq:Bounds}
				\frac{3(\max(p_{j_1}+p_{j_2}-1,0)-p_{j_1}p_{j_2} )}{(2-p_{j_1})(2-p_{j_2})} 
				\leq 
				\rho_S(X_{j_1},X_{j_2}) 
				\leq \frac{3 p_{j_1}(1-p_{j_2})}{(2-p_{j_1})(2-p_{j_2})},
			\end{equation*}
			
		\end{enumerate}
	\end{corollary}
	
	\begin{proof}
		From the stochastic representation in  \eqref{stoch_repre_X} of Theorem \ref{thm_stoch_repre_X}, we have
		\begin{equation*}
			X_j=\begin{cases}
				Z_{1,j} \,\,\, &\text{if } I_j=1\\
				Z_{0,j} \,\,\, &\text{if } I_j=0
			\end{cases}
			, \quad j \in \{1,\dots,d\}.
		\end{equation*}
		Thus, $X_j = Z_{I_j,j}$, $j \in \{1,\dots,d\}$, and by conditioning on $\II$, 
		the expectation of a function $\varphi$ of $\XX$ is given by
		\begin{align}
			E[\varphi(X_1,\dots,X_d) ] 
			& =
			\sum_{\boldsymbol{i} \in \{0,1\}^d} 
			f_{\II}(\boldsymbol{i})
			E[\varphi(Z_{i_1,1},\dots,Z_{i_d,d})], \label{eq:BeautifulFormula}
		\end{align}
		assuming that the expectations exist. 
		\begin{enumerate}
			\item By replacing \eqref{eq:ConvexRepresentationRays} in \eqref{eq:BeautifulFormula}, \eqref{eq:ephi} follows.

			\item By choosing $\varphi(\xx) = \boldsymbol{1}\{ x_1+\dots+x_d \leq y \}$, where $\boldsymbol{1}\{ A \} = 1$, if $A$ is true, and $\boldsymbol{1}\{ A \} = 0$, otherwise,
			\eqref{eq:FSwithExtremalBeautiful} follows.

			\item
			Covariance between $X_{j_1}$ and $X_{j_2}$ in \eqref{covariance}, $1 \leq j_1 < j_2 \leq d$, follows by considering $\varphi(\xx)=x_{j_1}x_{j_2}$ in \eqref{eq:BeautifulFormula}.
			\item Spearman's rho $\rho_S$ for any pair of continuous rvs $(X_{j_1}, X_{j_2})$, $1 \leq j_1 < j_2 \leq d$, is given by
			\begin{equation} \label{eq:RhoSpearman1}
				\rho_S(X_{j_1},X_{j_2}) = \rho_P(U_{j_1},U_{j_2}) 
				= \frac{E[U_{j_1}U_{j_2}] - E[U_{j_1}]E[U_{j_2}]}
				{\sqrt{Var(U_{j_1}) Var(U_{j_2})}},
			\end{equation}
			where $\rho_P$ indicates the Pearson's correlation and $Var(U_{j_1}) = Var(U_{j_2}) = \frac{1}{12}$.
			Using the representation in \eqref{eq:representation-u},  the expression for $E[U_{j_1}U_{j_2}]$, $1 \leq j_1 < j_2 \leq d$, is 
			\begin{align*}
				E[U_{j_1}U_{j_2}] 
				=
				E[U_{0,j_1}^{1-{p_{j_1}}} U_{0,j_2}^{1-p_{j_2}} U_{1,j_1}^{I_{j_1}} U_{1,j_2}^{I_{j_2}}] 
				=
				E[U_{0,j_1}^{1-p_{j_1}}] E[U_{0,j_2}^{1-p_{j_2}}] E[U_{1,j_1}^{I_{j_1}} U_{1,j_2}^{I_{j_2}}],
			\end{align*}
			where
			\begin{equation*}
				E[U_{0,j_1}^{1-p}] E[U_{0,j_2}^{1-p}]
				=
				\frac{1}{2-p_{j_1}} \frac{1}{2-p_{j_2}}
			\end{equation*}
			and
			\begin{align}
				\label{eq:EU1U2}
				E[U_{1,j_1}^{I_{j_1}} U_{1,j_2}^{I_{j_2}}]
				& = 
				\sum_{(i_{j_1}, i_{j_2} )\in \{0,1\}^2}
				f_{\II}^{(j_1, j_2)} (i_{j_1},i_{j_2})
				E[ U_{1,j_1}^{i_{j_1}}] E[U_{1,j_2}^{i_{j_2}}],
			\end{align}
			where $f_{\II}^{(j_1, j_2)}$ denotes the $(j_1,j_2)$-marginal pmf of $\II=(I_1,\dots,I_d)$.
			Finally, replacing \eqref{eq:EU1U2} in \eqref{eq:RhoSpearman1}, the expression of the Spearman's rho is given by
			\begin{align*} 
				\rho_S(X_{j_1},X_{j_2}) 
				&=
				\frac{12 f_{\II}^{(j_1, j_2)}(0,0)
					+ 6 (f_{\II}^{(j_1, j_2)}(0,1) + f_{\II}^{(j_1, j_2)}(1,0))
					+3 f_{\II}^{(j_1, j_2)}(1,1)}{(2-p_{j_1})(2-p_{j_2})} 
				- 3 \notag \\
				& 
				=\frac{3(f_{\II}^{(j_1, j_2)}(1,1)-p_1p_2) }{(2-p_{j_1})(2-p_{j_2})}  \notag\\
				& =
				\frac{3 \, Cov(I_{j_1},I_{j_2})}{(2-p_{j_1})(2-p_{j_2})},
			\end{align*}
			for $1 \leq j_1 < j_2 \leq d$. 
			
			Assume $p_{j_1} \leq p_{j_2}$. 
			In Section 3.2 of \cite{fontana2018representation},  the authors found the bounds for the covariance between $(I_{j_1}, I_{j_2})$ for any pair of Bernoulli variables. 
			The maximum value is $Cov(I_{j_1},I_{j_2})=p_{j_1}(1-p_{j_2})$,  while the minimum value is $Cov(I_{j_1},I_{j_2})=\max(p_{j_1}+p_{j_2}-1,0) - p_{j_1}p_{j_2}$. In the first case, $E[I_{j_1}I_{j_2}] = \Pr(I_{j_1}=1,I_{j_2}=1) = p_{j_1}$ and $Cov(I_{j_1},I_{j_2})=p_{j_1}(1-p_{j_2})$.  
			In the second case, $E[I_{j_1}I_{j_2}] = \Pr(I_{j_1}=1,I_{j_2}=1) = \max(p_{j_1}+p_{j_2}-1,0)$ and $Cov(I_{j_1},I_{j_2}) = \max(p_{j_1}+p_{j_2}-1,0) - p_{j_1}p_{j_2}$. Therefore, we have the following  bounds for Spearman's rho of any pair of continuous rvs
			\begin{equation*}
				\frac{3(\max(p_{j_1}+p_{j_2}-1,0)-p_{j_1}p_{j_2} )}{(2-p_{j_1})(2-p_{j_2})} 
				\leq 
				\rho_S(X_{j_1},X_{j_2}) 
				= 
				\rho_P(U_{j_1},U_{j_2}) 
				\leq \frac{3 p_{j_1}(1-p_{j_2})}{(2-p_{j_1})(2-p_{j_2})},
			\end{equation*}
			for $1 \leq j_1 < j_2 \leq d$. 
		\end{enumerate}
	\end{proof}
	The relation \eqref{eq:ephi} in particular holds for $\phi(S^{(\XX)})=\varphi(\XX)$  and implies that the bounds of $E[\phi(S^{(\XX)})]$, for any $\phi$ for which the expectation exists,   can be found by enumerating their values on the sums $S^{(\RRR_k)},$ $k \in \{1,\ldots, n_p^{\BBB}\}$, and this is computationally feasible in low dimension. 
	
	
	We end this section considering two examples for $\mathcal{G}^{\pp}_d(F_1,\ldots,F_d)$, where the margins are discrete in the first one and the margins are exponential in the second one.

	\begin{example} [Discrete margins]
		\label{ex:DiscreteMargins}
		Let $\XX = (X_1,\ldots,X_d)$ be defined as a $d$-dimensional random vector, where, for every $j \in \{1,\ldots,d\}$, $X_j$ is a discrete random variable taking values on the set $A_j=\{0,1,\ldots,n_j\}$, $n_j \in \mathbb{N}$, and with cdf $F_j$.
		When the joint distribution of $\XX$ belongs to the class $\mathcal{G}^{\pp}_d(F_1,\ldots,F_d)$, $\XX$ admits the representation in \eqref{stoch_repre_X}. 
		Moreover, for all $j \in \{1,\ldots,d\}$, to derive the values of the pmf of the discrete rvs $Z_{i,j},\, i=0,1$, we observe that, for each $k \in A_j$, we have 
		\begin{equation*}
			f_{Z_{0,j}}(k) = \Pr(Z_{0,j}=k) = \Pr(F_j^{-1}(V_{0,j}) = k) = F_{V_{0,j}}(F_j(k))-F_{V_{0,j}}(F_j(k-1))
		\end{equation*} 
		and
		\begin{equation*}
			f_{Z_{1,j}}(k) = \Pr(Z_{1,j}=k) = \Pr(F_j^{-1}(V_{0,j}V_{1,j}) = k) = F_{V_{0,j}V_{1,j}}(F_j(k)) - F_{V_{0,j}V_{1,j}}(F_j(k-1)),
		\end{equation*} 
		where 
		\begin{equation*}
			F_{V_{0,j}}(u)=u^{\frac{1}{1-p_j}},\quad u\in(0,1),
		\end{equation*}
		
		\begin{equation*}
			F_{V_{0,j}V_{1,j}}(u)=\frac{u}{p_j}-\frac{1-p_j}{p_j}u^{\frac{1}{1-p_j}},\quad u\in(0,1),
		\end{equation*}
		and $F_j(-1):=0$.
		For all $j\in\{1,\ldots,d\}$, it follows that the expectations of the discrete rvs $Z_{i,j},\, i=0,1$, are
		\begin{equation*}
			E[Z_{0,j}]=n_j-\sum_{k=0}^{n_j-1} F_{V_0,j}(F_j(k))
		\end{equation*}
		and
		\begin{equation*}
			E[Z_{1,j}]=n_j-\sum_{k=0}^{n_j-1} F_{V_0V_1,j}(F_j(k)).
		\end{equation*}
		Finally, we consider the sum $S^{(\XX)}=X_1+\cdots+X_d$, which can be rewritten using the representation in \eqref{stoch_repre_X}, as 
		\begin{equation}\label{eq:RepresGFGMSdv2}
			S^{(\XX)}=\sum_{j=1}^d (1-I_j)Z_{0,j} + I_j Z_{1,j}. 
		\end{equation}
		From \eqref{eq:RepresGFGMSdv2}, the expression of the probability generating function (pgf) of $S^{(\XX)}$ is given by
		\begin{equation}
			\label{eq:PfgSdiscrete}
			\mathcal{P}_{S^{(\XX)}}(s)
			=
			\sum_{\ii \in \{0,1\}^d} 
			f_{\II}(\ii)
			\prod_{j=1}^d \mathcal{P}_{Z_{i_j,j}}(s),
			\quad 
			s \in [-1,1],
		\end{equation}
		where the pgf of $Z_{i_j,j}$ is 
		\begin{equation*}
			\mathcal{P}_{Z_{i_j,j}}(s)
			=
			E[s^{Z_{i_j,j}}]
			=
			\sum_{k=0}^{n_j} f_{Z_{i_j,j}}(k) s^k,
			\quad 
			s \in [-1,1],
		\end{equation*}
		for $i_j \in \{0,1\}$ and $j \in \{1,\dots,d\}$.
		One uses the Fast Fourier Transform (\texttt{FFT}) algorithm of \cite{cooley1965algorithm} to extract the values of the pmf of $S^{(\XX)}$ from its pgf in $\eqref{eq:PfgSdiscrete}$. Details about that efficient approach is explained in Chapter 30 of \cite{cormen2009introductiona}. See also \cite{embrechts1993applications} for FFT applications in actuarial science and quantitative risk management.
		This procedure is illustrated in Example \ref{ex:final}, within Section \ref{NumILL}.
		\hfill \qed
		
		
		
	\end{example}

	\begin{example} [Exponential margins]
		\label{ex:ExponMargins}
		Assume that $\XX\in \mathcal{G}^{\pp}_d(F_1,\ldots,F_d)$, where $F_j$ is the cdf of an exponential distribution with mean $\frac{1}{\lambda_j}$, $j \in \{1,\ldots, d\}$. 
		In \cite{blier2024new}, they present an alternative stochastic representation of $\XX$ equivalent to \eqref{stoch_repre_X} that allows finding an analytical expression for the distribution of the sum of the components of the random vector $\XX$. 
		This other stochastic representation is
		\begin{equation}\label{stoch_repre_W}
			\XX=\WW_{1}+\II \WW_{2},
		\end{equation}
		where the components $W_{1,j}$ within $\WW_1$ and $W_{2,j}$ within $\WW_2$ are independent, with $W_{1,j}$ following an exponential distribution with mean $\frac{1-p_j}{\lambda_j}$ and $W_{2,j}$ also following an exponential distribution but with mean $\frac{1}{\lambda_j}$, for $j \in \{1,\ldots, d\}$. 
		The random vectors $\II$, $\WW_1$ and $\WW_2$ are independent.

		From \eqref{stoch_repre_W}, Pearson's coefficient becomes
		\begin{equation} \label{eq:PearsonRhoExpMargin}
			\rho_P(X_{j_1},X_{j_2}) 
			= 
			\frac{Cov(I_{j_1},I_{j_2})E[W_{2,j_1}]E[W_{2,j_2}]}{\sqrt{Var(X_{j_1})Var(X_{j_2})}} =
			Cov(I_{j_1},I_{j_2})
		\end{equation}
		since $E[W_{2,j_1}] = \frac{1}{\lambda_{j_1}}$, $E[W_{2,j_2}] = \frac{1}{\lambda_{j_2}}$, $Var(X_{j_1}) = \frac{1}{\lambda_{j_1}^2}$, and $Var(X_{j_2}) = \frac{1}{\lambda_{j_2}^2}$, $1 \leq j_1 < j_2 \leq d$.
		
		Considering now the sum $S^{(\XX)}$ of the components of $\XX$, the representation in \eqref{stoch_repre_W} allows us to express it as follows:
		\begin{equation}\label{eq:RepresGFGMSdv}
			S^{(\XX)} \overset{\mathcal{L}}{=} 
			W_{1,1} + I_1 W_{2,1}
			+ \ldots +
			W_{1,d} + I_d W_{2,d} = 
			\sum_{j=1}^d W_{1,j} + \sum_{j=1}^d I_j W_{2,j}.
		\end{equation}
		To identify the distribution of $S^{(\XX)}$, we find from \eqref{eq:RepresGFGMSdv}, the expression of the Laplace-Stieltjes transform (LST) of $S^{(\XX)}$, denoted by $\mathcal{L}_{S^{(\XX)}}$ and given by
		\begin{equation} \label{eq:GFGMLSTofSd}
			\mathcal{L}_{S^{(\XX)}}(t)
			= 
			\prod_{j=1}^d \mathcal{L}_{W_{1,j}}(t)
			\left(
			\sum_{\ii \in \{0,1\}^d} f_{\II}(\ii) 
			\prod_{j=1}^d \left( \mathcal{L}_{W_{2,j}}(t) \right)^{i_j}
			\right)
			, 
			\quad t \geq 0.  
		\end{equation}
		Using techniques explained in \cite{willmot2007class} and \cite{cossette2013multivariate}, the LST of $S$ admits the representation given by
		\begin{equation*} \label{eq:GFGMLSTofSdNo2}
			\mathcal{L}_{S^{(\XX)}}(t) 
			=
			\sum_{k=0}^{\infty} \eta_k \left( \frac{\beta}{\beta+t} \right)^{d+k}
			, 
			\quad t \geq 0    
		\end{equation*}
		where $\beta = \max \left(\lambda_1, \dots, \lambda_d, \frac{\lambda_1}{1-p_1}, \dots, \frac{\lambda_d}{1-p_d} \right)$, $\eta_k \geq 0$ for $k \geq 0$, 
		and $\sum_{k=0}^{\infty}\eta_k = 1$.
		From the expression of its LST in \eqref{eq:GFGMLSTofSd}, it follows that the rv $S$ follows a mixed Erlang distribution.   
		An application of \eqref{eq:GFGMLSTofSd} is provided in Example \ref{mainex}. 
		Details about the computation of the sequence of probabilities $\{\eta_k, k \in \mathbb{N}_0\}$ are explained in \cite{willmot2007class} and \cite{cossette2013multivariate}.   
		\hfill \qed
	\end{example}
	
	In Example \ref{ex:ExponMargins}, we find the distribution of $S$, which comes with an analytical expression for $F_S$ because we assume that the margins are exponential. In most cases, if the margins do not belong to a class of distributions closed under convolution, one must resort to numerical approximations. One of them is to use discretization techniques as explained in Section 4.3 of \cite{blier2023risk} jointly with the method explained in Example \ref{ex:DiscreteMargins}. 
	
	\section{Main result: Minimal convex sums under GFGM dependence} \label{FGMcop}

	Based on the results recalled and highlighted  in Section~\ref{prelim}, we provide in this section the geometrical structure embedded within the class of sums of components of random vectors whose joint distribution belongs to the class of $d$-variate GFGM copulas with a common parameter $p$, that is the class $\mathcal{C}_d^{p}$.
	Then, we are able to present a new result about convex order, as stated in Theorem~\ref{thm_sum_U}.
	
	Let us begin with the following proposition, which is a direct consequence of the one-to-one relationship given in \eqref{iff}.
	
	\begin{proposition} \label{prop_corresp_exch}
		Let $\II \in \mathcal{B}_d(p)$ and let $S^{(\II)}=\sum_{j=1}^d I_j$. Then, there exists an exchangeable Bernoulli random vector $\II^e \in \mathcal{E}_d(p)$ such that $S^{(\II^e)}=\sum_{j=1}^d I^e_j \in \mathcal{D}_d(dp)$ has the same distribution as $S^{(\II)}$.
	\end{proposition}
	\begin{proof}
		Given $\II \in \mathcal{B}_d(p)$, the sum of the components $S^{(\II)}=\sum_{j=1}^d I_j$ is a random variable that takes values in the set $\{0,1,\ldots,d\}$ and whose mean is $\mathrm{E}[S^{(\II)}]=dp$, that is $S^{(\II)} \in \mathcal{D}_d(dp)$. Therefore, from \eqref{iff}, it follows that there exists one exchangeable Bernoulli random vector $\II^e \in \mathcal{E}_d(p)$ such that $S^{(\II^e)} = \sum_{j=1}^d I^e_j \overset{\mathcal{L}}{=} S^{(\II)}$.
	\end{proof}

	Let $\mathcal{S}_d^{p}$ denote the class of sums of components of random vectors with multivariate distributions in $\mathcal{C}_d^p$.
	Let us indicate with $S^{(\UU, \II)}=\sum_{j=1}^d U^{(\II)}_j$, where $\UU^{(\II)}$ is the random vector whose joint cdf is the copula $C$ associated to $\II\in \BBB_d(p)$, as defined in \eqref{eq:representation-u}. 
	In order to study the geometrical structure of $\mathcal{S}_d^{p}$, we need to investigate the correspondence between distributions belonging to this class and multivariate Bernoulli distributions of the class $\mathcal{B}_d(p)$.
	
	\begin{lemma} \label{lemma_sum_distrib}
		Let $\II,\II' \in \mathcal{B}_d(p)$ and let $S^{(\II)}=\sum_{j=1}^d I_j$ and $S^{(\II')}=\sum_{j=1}^d I'_j$. Let $C$ and $C'$ be the  GFGM(p) copulas corresponding to $\II$ and $\II'$, respectively. Finally, let $\UU$ and $\UU'$ be uniform random vectors with joint cdfs $C$ and $C'$, respectively. If $S^{(\II)} \overset{\mathcal{L}}{=} S^{(\II')}$, then $\sum_{j=1}^d U_j \overset{\mathcal{L}}{=} \sum_{j=1}^d U'_j$.
	\end{lemma}
	\begin{proof}
		Let $F$ and $F'$ be the cdfs of $\sum_{j=1}^d U_j$ and $\sum_{j=1}^d U'_j$, respectively.
		From \eqref{eq:representation-u}, we have the following stochastic representation:
		\begin{equation*}\label{FGM_stoc}
			U_j = U_{0,j}^{1-p} U_{1,j}^{I_j}, \quad j \in \{1,\ldots,d\},
		\end{equation*}
		where $\UU_0 = (U_{0,1},\ldots,U_{0,d})$ and $\UU_1 = (U_{1,1},\ldots,U_{1,d})$ are vectors of independent standard uniform random variables and $\UU_0$, $\UU_1$, and $\II=(I_{1},\ldots,I_{d})$  are independent. 
		Then, we have
		\begin{align}
			F(x) &= \Pr \bigg(\sum_{j=1}^d U_j \leq x \bigg) = \Pr \bigg(\sum_{j=1}^dU_{0,j}^{1-p}U_{1,j}^{I_j} \leq x \bigg)\notag\\
			& = \sum_{\ii \in \{0,1\}^d} \Pr \bigg(\sum_{j=1}^d U_{0,j}^{1-p}U_{1,j}^{i_j} \leq x \bigg) f_{\II}(\ii)  \notag\\
			&= \sum_{\ii \in \{0,1\}^d} \Pr \bigg(\sum_{j:i_j=1} U_{0,j}^{1-p}U_{1,j} + \sum_{j:i_j=0} U_{0,j}^{1-p} \leq x \bigg) f_{\II}(\ii),
			\quad  {x\in [0,d] }.\label{eq:F}
		\end{align}
		Since $\UU_{0}$ and $\UU_{1}$ are independent vectors of iid random variables, the distribution of the sum $\sum_{j:i_j=1} U_{0,j}^{1-p}U_{1,j} + \sum_{j:i_j=0} U_{0,j}^{1-p}$ does not depend on the position of the 1's in $\ii$, but only on the number of 1's, that is, on the sum of the components $\sum_{j=1}^d i_j$.
		Hence, the cdf in \eqref{eq:F} becomes
		\begin{equation}\label{eq:F2}
			\begin{split}
				F(x) = \sum_{k=0}^d \sum_{\ii \in \mathcal{X}_d^k} \Pr \bigg(\sum_{j=1}^k U_{0,j}^{1-p}U_{1,j} + \sum_{j=k+1}^d U_{0,j}^{1-p} \leq x \bigg)f_{\II}(\ii),
			\end{split}
		\end{equation}
		where $\mathcal{X}_d^k = \{\ii \in \{0,1\}^d: \sum_{j=1}^d i_j = k \}$ and we set $\sum_{j=1}^0 U_{0,j}^{1-p}U_{1,j}:=0$ and $\sum_{j=d+1}^d U_{0,j}^{1-p}:=0$.
		It follows that
		\begin{equation} \label{distrF}
			\begin{split}
				F(x) = \sum_{k=0}^d \Pr \bigg(\sum_{j=1}^k U_{0,j}^{1-p}U_{1,j} + \sum_{j=k+1}^d U_{0,j}^{1-p} \leq x \bigg) f_{S^{(\II)}}(k),
				\quad x \in [0,d].
			\end{split}
		\end{equation}
		Therefore, given that the first probability in the summation of \eqref{distrF} does not depend on $\ii$ for $k\in \{0,\ldots,d\}$, we conclude that $S^{(\II)} \overset{\mathcal{L}}{=} S^{(\II')}$ implies $F(x) = F'(x)$, for every $x \in [0,d]$.
	\end{proof}
	We can now prove the following theorem that characterizes $\mathcal{S}_d^{p}$.
	
	\begin{theorem}\label{geomSumCop}
		The class $\mathcal{S}_d^p$ is a convex polytope and its extremal points are the cdfs $F_{S^{(\UU, \EE_k)}}$ of the random variables $S^{(\UU, \EE_k)}=\sum_{j=1}^dU^{(\EE_k)}_j$, where 
		$\UU^{(\EE_k)} \overset{\mathcal{L}}{=} \UU_0 ^{{1} - {p}} \UU_1^{\boldsymbol{(\EE_k)}}$,
		and $\EE_k$ is an extremal point of $\mathcal{E}_d(p)$, for $k \in \{1,\ldots, n_p^{\mathcal{D}}\}$.
	\end{theorem}
	
	
		
		
		\begin{proof}
			Consider any  $S \in \mathcal{S}_d^p$, then there exists $\II \in \BBB_d(p)$ such that  $S=\sum_{i=1}^dU^{(\II)}_i$. 
			By Proposition~\ref{prop_corresp_exch}, there exists a unique $\II^e\in \EEE(p)$ with $\sum_{j=1}^d I^e_j \overset{\mathcal{L}}{=} \sum_{j=1}^dI_j$. 
			Thus by Lemma~\ref{lemma_sum_distrib}, $S \overset{\mathcal{L}}{=} \sum_{j=1}^d U_j^e$, where $\UU^e\in \mathcal{C}_d^p$ is the uniform random vector with copula corresponding to $\II^e$. In other words, $\UU^e$ admits the representation
			$\UU^e \overset{\mathcal{L}}{=} \UU_0 ^{{1} - {p}} \UU_1^{\II^e}$,
			where $\II^e\in \EEE(p)$. 
			Let  $F_S$ be the cdf of $S$. Then, we have
			\begin{equation*}
				\begin{split}
					F_S(x) &= \Pr \bigg(\sum_{j=1}^d U^e_j \leq x \bigg) = \Pr \bigg(\sum_{j=1}^dU_{0,j}^{1-p}U_{1,j}^{I^e_j} \leq x \bigg) = \sum_{\ii \in \{0,1\}^d} \Pr \bigg(\sum_{j=1}^d U_{0,j}^{1-p}U_{1,j}^{i_j} \leq x \bigg) f_{\II^e}(\ii) \\
					&= \sum_{k=1}^{n_p^{\mathcal{D}}} \lambda_k \sum_{\ii \in \{0,1\}^d} \Pr \bigg(\sum_{j=1}^d U_{0,j}^{1-p}U_{1,j}^{i_j} \leq x \bigg) f_{\EE_k}( \ii) 
					= \sum_{k=1}^{n_p^{\mathcal{D}}} \lambda_k \Pr \bigg(\sum_{j=1}^d U^{(\EE_k)}_j \leq x \bigg) \\
					&= \sum_{k=1}^{n_p^{\mathcal{D}}} \lambda_k F_{S^{(\UU, \EE_k )}}(x),
				\end{split}
			\end{equation*}
			where $F_{S^{(\UU, \EE_k)}}$ is the cdf of $S^{(\UU, \EE_k)}=\sum_{j=1}^d U^{(\EE_k)}_j$.
		\end{proof}
		
		From Theorem \ref{geomSumCop}, we can complete the 
		relationship between classes in \eqref{iff} as follows:
		\begin{equation*}\label{iff2}
			\mathcal{S}_d^p \leftrightarrow 	\mathcal{E}_d(p) \leftrightarrow \mathcal{D}_d(dp).
		\end{equation*}
		We now study the convex order of sums of the components of random vectors with joint distribution described by a GFGM$(p)$ copula.
		We first recall the definition of the convex order.
		\begin{definition}
			Given two random variables $X$ and $Y$ with finite means, $X$ is said to be smaller
			than $Y$ in the convex order (denoted $X \preceq_{cx}Y$) if
			$
			E[\phi(X)]\leq E[\phi(Y)],
			$
			for all real-valued convex functions $\phi$ for which the expectations exist.
		\end{definition}

		In the proof of the following Theorem \ref{thm_sum_U}, we need to recourse to the supermodular order that we recall below (see Definition 3.8.5 in \cite{kleiber2005muller}).
		A function $\varphi: \RR^d \to \RR$ is said to be supermodular if $\varphi(\xx)+\varphi(\yy) \leq \varphi(\xx\vee\yy)+\varphi(\xx\wedge\yy)$,
		where the operators  $\wedge$ and $\vee$  denote coordinatewise minimum and maximum respectively.
		\begin{definition}
			We say that $\UU$ is smaller than $\UU'$ under the supermodular order, denoted $\UU \preceq_{sm} \UU'$, if $E[\varphi(\UU)] \leq E[\varphi(\UU')] $ for all supermodular functions $\varphi$, given that the expectations exist. 
		\end{definition}
		\begin{theorem}\label{thm_sum_U}
			Let $\II,\II' \in \BBB_d(p)$ be such that $\sum_{j=1}^d I_j \preceq_{cx} \sum_{j=1}^d I'_j$. Let $C$ and $C'$ be the GFGM copulas associated to $\II$ and $\II'$ and let $\UU$ and $\UU'$ be uniform random vectors with joint cdf $C$ and $C'$, respectively. Then,
			$
			\sum_{j=1}^d U_j \preceq_{cx} \sum_{j=1}^d U'_j.
			$
		\end{theorem}
		\begin{proof}
			By Proposition~\ref{prop_corresp_exch}, there exist two exchangeable Bernoulli random vectors $\II^e$ and $\II^{e'}$ of the class $\mathcal{E}_d(p)$ such that $\sum_{j=1}^d I_j \overset{\mathcal{L}}{=} \sum_{j=1}^d I^e_j$ and $\sum_{j=1}^d I'_j \overset{\mathcal{L}}{=} \sum_{j=1}^d I^{e'}_j$. Therefore, by Lemma~\ref{lemma_sum_distrib}, we have
			\begin{equation} \label{equality_distrib}
				\sum_{j=1}^d U_j \overset{\mathcal{L}}{=} \sum_{j=1}^d U^e_j \quad\text{and}\quad \sum_{j=1}^d U'_j \overset{\mathcal{L}}{=} \sum_{j=1}^d U^{e'}_j,
			\end{equation}
			where $\UU^e$ and $\UU^{e'}$ are uniform random vectors with joint distributions given by the GFGM copulas associated to $\II^e$ and $\II^{e'}$, respectively.
			Moreover, since $\sum_{j=1}^d I_j \preceq_{cx} \sum_{j=1}^d I'_j$ by hypothesis, then $\sum_{j=1}^d I^e_j \preceq_{cx} \sum_{j=1}^d I^{e'}_j$. However, as a consequence of results in Section 3 of \cite{frostig2001comparison}, the following double implication holds:
			\begin{equation*}
				\sum_{j=1}^d I^e_j \preceq_{cx} \sum_{j=1}^d I^{e'}_j \iff \II^e \preceq_{sm} \II^{e'}.
			\end{equation*}
			Furthermore, by Theorem~4.2 of \cite{blier2022stochastic}, $\II^e \preceq_{sm} \II^{e'}$ implies $\UU^e \preceq_{sm} \UU^{e'}$. Since, given a convex function $\phi \colon \mathbb{R} \to \mathbb{R}$, the function $\psi(\xx)=\phi(x_1+\ldots+x_d)$ is supermodular, $\II^e \preceq_{sm} \II^{e'}$ implies, in particular, $\sum_{j=1}^d U^e_j \preceq_{cx} \sum_{j=1}^d U^{e'}_j$. Finally, by the equality in distribution in \eqref{equality_distrib}, we have $\sum_{j=1}^d U_j \preceq_{cx} \sum_{j=1}^d U'_j$.
		\end{proof}
		
		Theorem \ref{thm_sum_U} obviously holds for FGM copulas by setting $p=\tfrac{1}{2}$. 

		\subsection{Class $\GGG_d^p(F)$ of distributions}
		
		Let us introduce the class $\GGG_d^p(F)$ of joint cdfs with a copula in the class $\mathcal{C}_d^p$  and with the same marginal cdfs $F$. Consider $\XX \in \GGG_d^p(F)$. From  Theorem \ref{thm_stoch_repre_X}, there exists $\II \in \BBB_d(p)$ such that $\XX$ is built from the Bernoulli random vector $\II$, according to the stochastic representation \eqref{stoch_repre_X}. 
		Using this representation, we can generalize Lemma \ref{lemma_sum_distrib} as follows.

		\begin{lemma} \label{lemma_sum_distrib2}
			Let $\II,\II' \in \BBB_d(p)$ be such that $S^{(\II)}=\sum_{j=1}^d I_j$ and $S^{(\II')}=\sum_{j=1}^d I'_j$ are equal in distribution. 
			Let $\XX, \XX'\in \GGG_d^p(F)$ be the random vectors corresponding to $\II$ and $\II'$, respectively. 
			Then, $\sum_{j=1}^d X_j$ and $\sum_{j=1}^d X'_j$ have the same distribution.
		\end{lemma}
		
		\begin{proof}
			Similarly to the proof of Lemma \ref{lemma_sum_distrib}, Lemma \ref{lemma_sum_distrib2} is proved observing that 
			$\ZZ_0$ and $\ZZ_1$ have independent and identically distributed components.
		\end{proof}
		
		Again, the stochastic representation in \eqref{stoch_repre_X} helps us to find the following generalization of Theorem \ref{geomSumCop}.
		
		\begin{theorem}\label{geomSumCop2}
			The class $\mathcal{S}_d^p(F)$ of distributions of sums of components of vectors with distribution in $\GGG_d^p(F)$ is a convex polytope and its extremal points are the distributions  of  $S^{(\XX, \EE_k)}=\sum_{j=1}^dX^{(\EE_k)}_j$ where 
			\begin{equation}\label{stoch_repre_X_extremal}
				\XX^{(\EE_k)}=(\boldsymbol{1}-\EE_k)\ZZ_0+\EE_k \ZZ_1,
			\end{equation}
			where $\EE_k$ is an extremal point of $\mathcal{E}_d(p)$, for $k \in \{1,\ldots, n_p^{\mathcal{D}}\}$.
		\end{theorem}
		
		\begin{proof}
			The proof is similar to the one of Theorem \ref{geomSumCop}.
		\end{proof}
		
		Theorem \ref{geomSumCop2} completes the relationship in \eqref{iff2} as follows:
		\begin{equation*}\label{iff3}
			\mathcal{S}_d^p(F) \leftrightarrow \mathcal{S}_d^p \leftrightarrow 	\mathcal{E}_d(p) \leftrightarrow \mathcal{D}_d(dp).
		\end{equation*}
		From this relationship, it follows that the number of extremal points in $\mathcal{S}_d^p(F)$ is $n^{\mathcal{D}}_p$, that is significantly lower than the number of extremal points in $\mathcal{G}_d^p(F)$ and they are analytical. Therefore we can also find them in high dimensions.

		Finally, we have the following generalization of Theorem \ref{thm_sum_U}.
		
		\begin{theorem} \label{prop:cxOrder1}
			Let $\II,\II' \in \BBB_d(p)$ and let $\XX\in \GGG_d^p(F)$ and $\XX'\in \GGG_d^p(F)$, for any cdf $F$, be respectively built from $\II$ and $\II'$, as in (\ref{stoch_repre_X}).
			Then, 
			\begin{equation*}
				\sum_{j=1}^d I_j \preceq_{cx} \sum_{j=1}^d I'_j \implies \sum_{j=1}^d X_j \preceq_{cx} \sum_{j=1}^d X'_j.
			\end{equation*}
		\end{theorem}
		
		\begin{proof}
			The proof of this Theorem is along the same lines at the one of Theorem \ref{thm_sum_U}. 
		\end{proof}

		We conclude this section by considering the two examples with exponential and discrete margins previously discussed in Example~\ref{ex:DiscreteMargins} and Example~\ref{ex:ExponMargins} but here in the special case of identically distributed risks. 
		
		\setcounter{example}{0}
		\begin{example}[Discrete margins, \textbf{continued}]
			For the case of identically distributed discrete margins, we obtain the following expression for the pmf of $S^{(\XX)}$:
			\begin{equation}
				f_{S^{(\XX)}}(k)
				=
				\sum_{j=0}^d f_{S^{(\II)}}(j) f_{Z_0}^{*(d-j)} * f_{Z_1}^{*j}(k),
				\quad
				k \in \mathbb{N},
				\label{eq:FSdiscrete}
			\end{equation}
			where * denotes the convolution product and $f^{*j}$ denotes the $j$-fold convolution product of the probability density function measure $f$ with itself.
			It follows from \eqref{eq:PfgSdiscrete} that the probability generating function (pgf) of $S$ is given by
			\begin{equation}
				\mathcal{P}_{S^{(\XX)}}(s)
				=
				\sum_{j=0}^d f_{S^{(\II)}}(j) 
				\mathcal{P}_{Z_0}^{d-j}(s) \mathcal{P}_{Z_1}^{j}(s),
				\quad
				s \in [0,1],
				\label{eq:pgfSdiscrete}
			\end{equation} 
			where $\mathcal{P}_{Z_0}$ and
			$\mathcal{P}_{Z_1}$ are the pgf of $Z_0$ and
			$Z_1$, respectively. 
			As previously mentioned, it is more convenient to use the Fast Fourier Transform (\texttt{FFT}) algorithm to extract the values of the pmf of $S^{(\XX)}$ from the pgf in \eqref{eq:pgfSdiscrete} rather than finding those values directly from \eqref{eq:FSdiscrete}. The procedure is illustrated in Example \ref{ex:final} provided in Section~\ref{NumILL}.
		\end{example}
		\begin{example}[Exponential margins, \textbf{continued}]
			Firstly, assume $\XX \in \mathcal{G}^{p}_d(Exp(\lambda))$, that is $X_j \sim Exp(\lambda)$ for every $j \in \{1,\dots, d\}$, with GFGM($p$) copula. The LST of $S^{(\XX)}$ in \eqref{eq:GFGMLSTofSd} becomes
			\begin{equation*} 
				\mathcal{L}_{S^{(\XX)}}(t) =   \left( \mathcal{L}_{W_{1}}(t) \right) ^d
				\left(
				\sum_{j=0}^d f_{S^{(\II)}}(j) 
				\left( \mathcal{L}_{W_{2}}(t) \right)^{j}
				\right)
				, 
				\quad t \geq 0,  \label{eq:GFGMLSTofSdwithN}
			\end{equation*}
			where $S^{(\II)} = \sum_{j=1}^d I_j$. By identification of the LST transform, we conclude
			\begin{equation*} \label{eq:FSextremalGFGM}
				F_{S^{(\XX)}}(x)
				=f_{S^{(\II)}}(0)F_{W_{1}}^{*d}(x)+
				\sum_{j=1}^d f_{S^{(\II)}}(j) F_{W_{1}}^{*d}
				* F_{W_{2}}^{*j}(x),
				\quad
				x \geq 0,
			\end{equation*}
			where $F_{W_{1}}^{*d}$ corresponds to the cdf of an $Erlang(d, \frac{\lambda}{1-p})$ distribution and $F_{W_{2}}^{*j}$ corresponds to the cdf of an $Erlang(j, {\lambda})$ distribution for $j\in \{1,\ldots, d\}$.
		\end{example}

		\subsection{Minimal convex sums} \label{sec:min_cx_sum}
		The solution to the problem of finding the distribution of the vectors with minimal convex sums is a corollary of 
		our main theorems, Theorem \ref{thm_sum_U} and Theorem \ref{prop:cxOrder1}.  Before introducing it, we need some notations.
		Following Definition 3.4 of \cite{puccetti2015extremal}, we say that  $\XX\in\FFF$ is a $\Sigma_{cx}$-smallest element in a class of distributions $\FFF$ if, for all $\XX'\in \FFF$,
		$\sum_{j=1}^d X_j\preceq_{cx}\sum_{j=1}^d X'_j.$
		A $\Sigma_{cx}$-smallest element in a Fr\'echet class does not always exist, see Example 3.1 of \cite{bernard2014risk}. 
		However, the authors of \cite{hu1999dependence} found the distribution of the exchangeable Bernoulli random vector that is the $\Sigma_{cx}$-smallest element in the class of exchangeable Bernoulli pmfs $\mathcal{E}_d(p)$.
		Since there is a one-to-one correspondence between $\mathcal{E}_d(p)$ and $\mathcal{D}_d(dp)$, see (\ref{iff}), and the sums of the components of a random vector with pmf in $\BBB_d(p)$ are rvs with distribution in $\mathcal{D}_d(dp)$, a $\Sigma_{cx}$-smallest element always exists in the class $\BBB_d(p)$.  
		Actually, for each $p \in (0,1)$, we can also build non-exchangeable $\Sigma_{cx}$-smallest elements of $\BBB_d(p)$ following Theorem 5.2 of \cite{fontana2024high}. In the proof of Lemma 3.1 of \cite{bernard2017robust}, the authors provide a way to construct a random variable with $\Sigma_{cx}$-smallest pmf.
		Let $\II \in \mathcal{B}_d(p)$ and let $\UU$ be the corresponding uniform random vector with GFGM$(p)$ copula. 
		The following corollary is a straightforward but important consequence of Theorem~\ref{thm_sum_U} and Theorem~\ref{prop:cxOrder1}. 
		
		\begin{corollary}
			\label{corol:SmallestElement}
			Let  $\II\in \mathcal{B}_d(p)$ be a $\Sigma_{cx}$-smallest element.  
			\begin{enumerate}
				\item Let $\UU$ be a uniform random vector whose joint cdf is the GFGM$(p)$ copula corresponding to $\II$.  Then,  $\UU$ is a $\Sigma_{cx}$-smallest element in $ \mathcal{C}_d^p$. 
				\item Let  $\XX \in \GGG^p_d(F)$ with joint cdf defined with the GFGM$(p)$ copula corresponding to $\II$. Then, $\XX$ is a $\Sigma_{cx}$-smallest element in $\GGG^p_d(F)$.
			\end{enumerate}
		\end{corollary}
		Consequently, distributions, for which the lower bounds of a convex functional are reached, are built using a $\Sigma_{cx}$-smallest element of $\mathcal{B}_d(p)$. 
		Obviously, using the upper Fr\'echet bound of $\mathcal{B}_d(p)$, we build the distributions of vectors with maximal convex sums.

		\section{Sharp bounds for risk measures} \label{sec:SharpBounds}

		As a consequence of Corollary \ref{corollary3.1}, to derive sharp bounds for risk measures in the class $\mathcal{S}_d^{\pp}(F_1,\ldots, F_d)$, we can proceed by enumerating their values on the extremal points. This is computationally expensive because the number of extremal points $n_p^{\mathcal{B}}$ explodes and becomes larger, as highlighted by the authors of \cite{fontana2024high}. In this section, we illustrate how this problem can be solved by finding sharp bounds for risk measures in the classes $\mathcal{C}_d^p$ and $\GGG^p_d(F_X)$, for any $p \in (0,1)$.

		As a motivation, we start with an example showing that the assumption of common margins in Lemma \ref{lemma_sum_distrib2} and Theorem \ref{prop:cxOrder1} is necessary.
		\begin{example}
			Consider the class $\mathcal{G}_3^{2/5}(F_1, F_2, F_3)$, where $F_1$, $F_2$, and $F_3$ are the discrete cdfs provided in Table \ref{contexpmf}.
			\begin{table}[hbt!]
				\centering
				\begin{tabular}{c|ccc}
					\toprule
					$k$ & $F_1(k)$ & $F_2(k)$ & $F_3(k)$ \\ 
					\midrule
					0 & 0.1 & 0.1 & 0.8 \\ 
					1 & 0.2 & 0.4 & 1.0 \\ 
					2 & 0.3 & 0.7 & 1.0 \\ 
					3& 1.0 & 1.0 & 1.0 \\ 
					\bottomrule
				\end{tabular}
				\caption{\emph{Marginal cdfs.}}
				\label{contexpmf}
			\end{table}
			Let $\XX$, $\XX'$, and $\XX'' \in \mathcal{G}_3^{2/5}(F_1, F_2, F_3)$ with copulas respectively defined by the Bernoulli rvs $\II$, $\II'$, and $\II''$ respectively distributed according $\ff$, $\ff'$, and $\ff'' \in \BBB_3(\tfrac{2}{5})$ given in Table \ref{Bernoulli}.
			\begin{table}[hbt!]
				\centering
				\begin{tabular}{c|cccccccc}
					\toprule
					$\ii $& (0,0,0) & (1,0,0) & (0,1,0) & (1,1,0) & (0,0,1) & (1,0,1) & (0,1,1) & (1,1,1) \\ 
					\midrule
					$\ff$ & 0 & $\tfrac{1}{5}$ & $\tfrac{1}{5}$ & $\tfrac{1}{5}$ & $\tfrac{2}{5}$ & 0 & 0 & 0 \\ 
					$\ff'$ & $\tfrac{1}{5}$ & 0 & $\tfrac{2}{5}$ & 0 & 0 & $\tfrac{2}{5}$ & 0 & 0 \\
					$\ff''$ & 0 & $\tfrac{2}{5}$ & $\tfrac{1}{5}$ & 0 & $\tfrac{1}{5}$ & 0 & $\tfrac{1}{5}$ & 0 \\
					\bottomrule
				\end{tabular}
				\caption{\emph{Bernoulli pmfs.}}
				\label{Bernoulli}
			\end{table}
			Table \ref{contexsum} exhibits the values of pmfs of the discrete random variables $S^{(\XX)}=\sum_{j=1}^3X_j$, $S^{(\XX')}=\sum_{j=1}^3X'_j$, and $S^{(\XX'')}=\sum_{j=1}^3X''_j$; those values are computed using the pgf and the \texttt{FFT} algorithm as explained in Example \ref{ex:DiscreteMargins}.
			As one can see by looking at the support of the Bernoulli distributions in Table \ref{Bernoulli}, although $\II$ and $\II''$ have the same distribution of the sum, $S^{(\XX)}$ and $S^{(\XX'')}$ are not equal in distribution.
			Moreover, it is easy to show that $S^{(\II)}=\sum_{j=1}^3 I_j \preceq_{cx} S^{(\II')}=\sum_{j=1}^3 I'_j$, but $S^{(\XX)}\npreceq_{cx} S^{(\XX')}$, since $Var(S^{(\XX)})=2.0633\geq 1.8865= Var(S^{(\XX')})$.
			
		\end{example}
		
		\begin{table}[hbt!]
			\centering
			\begin{tabular}{c|cccccccccc}
				\toprule
				$k$ & 0 & 1 & 2 & 3 & 4 & 5 & 6 & 7 & 8 & 9 \\ 
				\midrule
				$f_{S^{(\XX)}}(k)$ & 0.0080 & 0.0338 & 0.0640 & 0.1328 & 0.2467 & 0.2592 & 0.2312 & 0.0242 & 0 & 0 \\
				$f_{S^{(\XX')}}(k)$ & 0.0032 & 0.0249 & 0.0602 & 0.1556 & 0.2636 & 0.2569 & 0.2004 & 0.0352 & 0 & 0 \\
				$f_{S^{(\XX'')}}(k)$ & 0.0029 & 0.0214 & 0.0549 & 0.1588 & 0.2798 & 0.2521 & 0.1976 & 0.0324 & 0 & 0 \\ 
				\bottomrule
			\end{tabular}
			\caption{\emph{Pmfs of $S^{(\XX)}$ and $S^{(\XX')}$.}}
			\label{contexsum}
		\end{table}
		
		We consider two convex risk measures, the widely used expected shortfall (ES), and the entropic risk measure. Then we consider, consistently with Regulation, the value-at-risk (VaR), which is not a convex measure. Below we recall the definition of these three measures of risk.
		
		\begin{definition}
			Let Y be a random variable representing a loss with finite mean. Then, the value-at-risk at level $\alpha \in (0,1)$ is given by
			\begin{equation*}
				\text{VaR}_{\alpha}(Y) = \inf\{ y \in\RR : \Pr(Y \leq y) \geq \alpha \}.
			\end{equation*}
		\end{definition}
		\begin{definition}\label{def:cxmeas}
			Let Y be a random variable representing a loss with finite mean. The expected shortfall at level $\alpha \in (0,1)$ is defined as
			\begin{equation}\label{EF}
				\text{ES}_{\alpha}(Y) = \frac{1}{1-\alpha} \int_{\alpha}^{1} \text{VaR}_{u}(Y) \mathrm{d}u.
			\end{equation}
		\end{definition}
		The expected shortfall defined in \eqref{EF} also admits the following representation
		\begin{equation*}
			\text{ES}_{\alpha}(Y) = \text{VaR}_{\alpha}(Y)+\frac{1}{1-\alpha} E[\max(Y- \text{VaR}_{\alpha}(Y),0)], \, for \,\,\, \alpha\in(0,1).
		\end{equation*}
		\begin{definition}
			The entropic risk measure is defined by 
			\begin{equation*}
				\Psi_{\gamma}(Y) = \frac{1}{\gamma} \log(E[e^{\gamma Y}]),
			\end{equation*}
			assuming that there exists a real number $\gamma_0>0$ such that $E[e^{\gamma Y}]$ is finite for $\gamma\in (0, \gamma_0)$.
		\end{definition}	
		Using the results from Corollary~\ref{corol:SmallestElement}, we can analytically find the lower bounds of the convex risk measures considered for exponential margins and discrete margins, respectively.
		Although the  $\VaR_{\alpha}$ is not convex we can find its bounds in $\mathcal{S}^p_d(F)$.
		The authors of \cite{fontana2021model} prove that the bounds of the $\VaR_{\alpha}$ in a class of univariate distributions that has a convex polytope structure are reached at the extremal points.  
		We consider a random vector $\XX$ with distribution in the class $\mathcal{G}_d^p(F)$, whose sum $S^{(\XX)}=X_1+\ldots+X_d$ have cdf in $\mathcal{S}_d^p(F)$. 
		Despite the number of extremal points of $\mathcal{G}_d^p(F)$ being $n_p^{\mathcal{B}}$, as a consequence of Theorem \ref{geomSumCop2} we can restrict our attention to the $n_p^{\mathcal{D}}$ extremal points of $\mathcal{S}_d^p(F)$. By computing the $\VaR_{\alpha}$ of the random variable $S^{(\XX,\EE_k)}$, for every $k \in \{1,\dots,n_p^{\mathcal{D}} \}$, it is possible to find maximum and minimum values that $\VaR_{\alpha}(X_1+\dots+X_d)$ can reach.

		\begin{remark}In some cases, the minimum $\VaR_{\alpha}$ is reached on the minimal $\Sigma_{cx}$-element of the polytope.
			In fact, from the proof of Theorem 3.A.4 in \cite{shaked2007stochastic}, it follows that, if $S \preceq_{cx} S'$, then there exists $\tilde{\alpha} \in (0,1)$ such that, $\text{VaR}_{\alpha}(S) \leq \text{VaR}_{\alpha}(S')$, for every $\alpha \in(\tilde{\alpha},1)$. 
			Therefore, if $\XX=(X_1,\ldots,X_d)$ is a $\Sigma_{cx}$-smallest element of its Fr\'echet class, then there exists $\tilde{\alpha} \in (0,1)$ such that, $\text{VaR}_{\alpha}(\sum_{j=1}^{d}X_j) \leq \text{VaR}_{\alpha}(\sum_{j=1}^{d}X'_j)$, for every $\alpha \in(\tilde{\alpha},1)$, for every random vector $\boldsymbol{X}'=(X'_1,\ldots,X'_d)$ of the same Fr\'echet class of $\XX$.\end{remark}

		We present below two numerical illustrations of our results.
		\begin{example}
			Consider the case $d=5$, $p=\tfrac{1}{2}$. Then, $dp=\tfrac{5}{2}$ and $j_1^{\vee}=2$, $j_2^{\wedge}=3$. The class $\mathcal{D}_5(\tfrac{5}{2})$ has $n^{\mathcal{D}}_{1/2}=9$ extremal points provided in Table~\ref{Table_ExtrPoints}.
			Let us compute value-at-risk, expected shortfall, and entropic risk measures of the extremal pmfs for $\alpha=0.8$ and $\gamma=0.1$. 
			Results are reported in Table~\ref{Table_ResultsBernoulli}. 
			The choice of $\alpha$ has been made to exhibit the case where the minimum $\VaR_{\alpha}$ is not the $\Sigma_{cx}$-smallest element of the class.  
			Table~\ref{Table_ResultsFGM} reports instead the same risk measures evaluated on the corresponding FGM copulas.  
			Notice that the bounds for the sums are at the extremal copulas corresponding to the upper Fr\'echet bound and to the $\Sigma_{cx}$-smallest Bernoulli pmfs for the convex measures, as proved in Corollary~\ref{corol:SmallestElement}. 
			The minimum $\VaR_{\alpha}$ in $\mathcal{C}_5^{1/2}$ is reached at $C_{\rr^{D}_7}$ while for the Bernoulli case it is at ${\rr^{D}_9}$. 
			This proves that the minimum $\VaR_{\alpha}$ of the sum of the components of $\textbf{U}$ in $\mathcal{C}_5^{1/2}$ is not inherited from the underlying Bernoulli pmf. 
			
			\begin{table}[hbt!]
				\centering
				\begin{tabular}{c|ccccccccc}
					\toprule
					$y$ & $\rr^{D}_1$ & $\rr^{D}_2$ & $\rr^{D}_3$ & $\rr^{D}_4$ & $\rr^{D}_5$ & $\rr^{D}_6$ & $\rr^{D}_7$ & $\rr^{D}_8$ & $\rr^{D}_9$ \\
					\midrule
					0 & $\tfrac{1}{6}$ & $\tfrac{3}{8}$ & $\tfrac{1}{2}$ & 0    & 0   & 0   & 0   & 0    & 0 \\
					1 & 0   & 0   & 0   & $\tfrac{1}{4}$ & $\tfrac{1}{2}$ & $\tfrac{5}{8}$ & 0   & 0    & 0 \\
					2 & 0   & 0   & 0   & 0    & 0   & 0   & $\tfrac{1}{2}$ & $\tfrac{3}{4}$ & $\tfrac{5}{6}$ \\
					3 & $\tfrac{5}{6}$ & 0   & 0   & $\tfrac{3}{4}$ & 0   & 0   & $\tfrac{1}{2}$ & 0    & 0 \\
					4 & 0   & $\tfrac{5}{8}$ & 0   & 0    & $\tfrac{1}{2}$ & 0   & 0   & $\tfrac{1}{4}$ & 0 \\
					5 & 0   & 0   & $\tfrac{1}{2}$ & 0    & 0   & $\tfrac{3}{8}$ & 0   & 0    & $\tfrac{1}{6}$ \\
					\bottomrule
				\end{tabular}
				\caption{\emph{Extremal pmfs of the class $\mathcal{D}_5(\tfrac{5}{2})$.}}
				\label{Table_ExtrPoints}
			\end{table}
			
			\begin{table}[hbt!]
				\centering
				\begin{tabular}{c|ccccccccc}
					\toprule
					& $\rr^{D}_1$ & $\rr^{D}_2$ & $\rr^{D}_3$ & $\rr^{D}_4$ & $\rr^{D}_5$ & $\rr^{D}_6$ & $\rr^{D}_7$ & $\rr^{D}_8$ & $\rr^{D}_9$ \\
					\midrule
					$\text{VaR}_{0.8}(S^{(\II)})$ & 3 & 4 & \underline{5} & 3 & 4   & \underline{5} & 3 & 4 & \squared{2}   \\
					$\text{ES}_{0.8}(S^{(\II)})$  & \squared{3} & 4 & \underline{5} & \squared{3} & 4   & \underline{5} & \squared{3} & 4 & 4.5 \\
					$\Psi_{0.1}(S^{(\II)})$      & 2.5584 & 2.6803 & \underline{2.8093} & 2.5362 & 2.6121 & 2.6927 & \squared{2.5125} & 2.5387 & 2.5667 \\
					\bottomrule 
				\end{tabular}
				\caption{\emph{Values of risk measures of the sum of the components of Bernoulli random vectors with the extremal probability mass functions of the class $\mathcal{D}_5(\tfrac{5}{2})$. The minimum values are squared and the maximum ones are underlined.}}
				\label{Table_ResultsBernoulli}
			\end{table}
			
			\begin{table}[hbt!]
				\centering
				\begin{tabular}{c|ccccccccc}
					\toprule
					& $C_{\rr^{D}_1}$ & $C_{\rr^{D}_2}$ & $C_{\rr^{D}_3}$ & $C_{\rr^{D}_4}$ & $C_{\rr^{D}_5}$ & $C_{\rr^{D}_6}$ & $C_{\rr^{D}_7}$ & $C_{\rr^{D}_8}$ & $C_{\rr^{D}_9}$ \\
					\midrule
					$\text{VaR}_{0.8}(S^{(\UU)})$ & 3.0308 & 3.3281 & \underline{3.4928} & 3.0158 & 3.1710 & 3.2636 & \squared{2.9729} & 3.0180 & 3.0476 \\
					$\text{ES}_{0.8}(S^{(\UU)})$  & 3.4627 & 3.7345 & \underline{3.8401} & 3.3641 & 3.5161 & 3.5846 & \squared{3.2753} & 3.3228 & 3.3477 \\
					$\Psi_{0.1}(S^{(\UU)})$ & 2.5210 & 2.5350 & \underline{2.5486} & 2.5181 & 2.5264 & 2.5345 & \squared{2.5153} & 2.5180 & 2.5207 \\ 
					\bottomrule
				\end{tabular}
				\caption{\emph{Values of risk measures of the sum of components of uniform vectors with joint cdf defined by FGM copulas corresponding to the extremal probability mass functions of the class $\mathcal{D}_5(\tfrac{5}{2})$. The minimum values are squared and the maximum ones are underlined.}}
				\label{Table_ResultsFGM}
			\end{table}

		\end{example}

		The following Example \ref{mainex} is the main example of application of our results. We consider a high dimensional portfolio of risks for six scenarios: three different GFGM($p$) dependence structures for two Fr\'echet classes,  with  exponential and discrete margins discussed in a theoretical setting in Examples \ref{ex:DiscreteMargins} and \ref{ex:ExponMargins}. 
		
		\begin{example}\label{mainex}

			Consider the classes $\GGG_{100}^p\left(Exp \left(\tfrac{1}{10}\right)\right)$ and $\GGG_{100}^p(F)$, where $F$ is the discrete cdf whose pmf is given by
			\begin{equation*}
				f(y)=
				\begin{cases}
					0.8 , & y = 0 \\
					0.2 [\big(\frac{y}{100}\big)^3 - \big(\frac{y-1}{100}\big)^3 ], & y \in \{1,\ldots,100\}
				\end{cases}.             
			\end{equation*}
			We consider three different cases of GFGM(p) dependencies for each class, that is we consider $\GGG_{100}^p \left(Exp \left(\tfrac{1}{10}\right)\right)$ and $\GGG_{100}^p(F)$, where each case is associated to a common $p \in \{ \tfrac{1}{3}, \tfrac{1}{2}, \tfrac{2}{3}\}$. 
			
			The bounds for the convex measures are reached at the distributions of the two classes $\GGG_{100}^p \left(Exp \left(\tfrac{1}{10}\right)\right)$ and $\GGG_{100}^p(F)$ corresponding to the minimal and maximal convex sums in $\mathcal{D}_{100}(100p)$, for $p=\tfrac{1}{3}, p=\tfrac{1}{2}$ and $p=\frac{2}{3}$. The minimal convex sum  is the pmf in $\mathcal{D}_{100}(100p)$ with support on the pair $(33, 34)$ when $p=\tfrac{1}{3}$, with support on the point $50$ when $p=\tfrac{1}{2}$, and support on the pair $(66, 67)$ when $p=\tfrac{2}{3}$.
			Table \ref{CXmeas_Bound} provides the sharp bounds for the convex risk measures with exponential and discrete margins, respectively.

			\begin{table}[hbt!]
				\centering
				\begin{tabular}{c|ccc|ccc}
					\toprule
					& \multicolumn{3}{c|}{$\GGG_{100}^p \left(Exp \left(\tfrac{1}{10}\right)\right)$} & \multicolumn{3}{c}{$\GGG_{100}^p(F)$} \\
					\midrule
					& $p=\frac{1}{3}$  & $p=\frac{1}{2}$ & $p=\frac{2}{3}$ & $p=\frac{1}{3}$  & $p=\frac{1}{2}$ & $p=\frac{2}{3}$ \\
					\midrule
					$\min ES_{0.95}$    & 1191.2742 & 1189.2721 & 1192.3324 & 2152.595 & 2122.718 & 2019.207 \\
					$\max ES_{0.95}$    & 1858.1846 & 1702.8444 & 1540.6192 & 2858.955 & 3448.241 & 4440.057 \\
					$\min \Psi_{0.001}$ & 1003.9212 & 1003.8215 & 1003.9237 & 1555.710 & 1551.957 & 1546.627 \\
					$\max \Psi_{0.001}$ & 1124.6343 & 1125.0510 & 1101.5259 & 1888.303 & 2216.540 & 2843.312 \\
					\bottomrule
				\end{tabular}
				\caption{\emph{Bounds of convex risk measures in the classes $\GGG_{100}^p \left(Exp \left(\tfrac{1}{10}\right)\right)$ and $\GGG_{100}^p(F)$.}}
				\label{CXmeas_Bound}
			\end{table}

			The $\VaR_{\alpha}$ is bounded by its evaluations on the extremal pmfs of the two classes $\mathcal{S}_{100}^p \left(Exp \left(\tfrac{1}{10}\right)\right)$ and $\mathcal{S}_{100}^p(F)$. When $d=100$, the number of extremal points $n_p^{\mathcal{D}}$ is lower than or equal to 2501, see Corollary 4.6 of \cite{fontana2021model}, and we find bounds by enumeration. 
			Table~\ref{Var-Bounds}
			provides the bounds for the $\VaR_{\alpha}$ in the abovementioned classes and also the analytical bounds for the whole Fr\'echet classes given in Equation (4) of \cite{bernard2017robust}. 
			We mention that the minimum $\VaR_{0.95}$ in the class $\GGG_{100}^p (F)$ is not reached at the distribution corresponding to the minimal convex pmf in $\mathcal{D}_{100}\left(\tfrac{200}{3}\right)$, whose $\VaR_{\alpha}$ is 1961. 
			
			\begin{table}[hbt!]
				\centering
				\begin{tabular}{c|ccc|ccc}
					\toprule
					& \multicolumn{3}{c|}{$Exp \left(\tfrac{1}{10}\right)$} & \multicolumn{3}{c}{$F$} \\
					\midrule
					& $p=\frac{1}{3}$  & $p=\frac{1}{2}$ & $p=\frac{2}{3}$ & $p=\frac{1}{3}$  & $p=\frac{1}{2}$ & $p=\frac{2}{3}$ \\
					\midrule
					Lower bound $\mathcal{F}_{100}(G)$ & 842.3299 & 842.3299 & 842.3299 & 1045.963 & 1045.963 & 1045.963 \\
					Min $\mathcal{G}^p_{100}(G)$ & 1149.7294 & 1147.0118 & 1150.2229 & 2016 & 1994 & 1960 \\
					Max $\mathcal{G}^p_{100}(G)$ & 1791.3283 & 1645.0538 & 1488.2312 & 2688 & 3258 & 4225 \\
					Upper bound $\mathcal{F}_{100}(G)$& 3995.7323 &3995.7323 &3995.7323 & 9606.61 & 9606.61 & 9606.61 \\
					\bottomrule
				\end{tabular}
				\caption{\emph{$\text{VaR}_{0.95}$ Bounds: $\mathcal{F}_{100}(G)$ is the Fr\'echet class with $100$ identically distributed risks,  $X_j \sim G$, for every $j \in \{1,\dots,100\}$, and $\mathcal{G}^p_{100}(G)$, where $G=Exp \left(\tfrac{1}{10}\right)$ or $G=F$.}}
				\label{Var-Bounds}
			\end{table}

			We conclude this example by considering Pearson's correlation of the exchangeable $\Sigma_{cx}$-smallest element $\XX^e_{cx}$ in $\GGG_{100}^{1/3}\left(Exp \left(\tfrac{1}{10}\right)\right)$ and Pearson's correlation matrix of a vector $\XX_{cx}$ corresponding to the Bernoulli $\Sigma_{cx}$-smallest element provided by the authors of \cite{fontana2024high} in Theorem 5.2.
			Using \eqref{eq:PearsonRhoExpMargin}, Pearson's correlation $\rho_P(X_{j_1}, X_{j_2})$ (denoted by $\rho_{j_1 j_2}$) is equal to $Cov(I_{j_1}, I_{j_1})$, for $1 \leq j_1 < j_2 \leq d$.
			
			We therefore have to find the covariance of the  exchangeable $\Sigma_{cx}$-smallest element and the covariance matrix of the $\Sigma_{cx}$-smallest element $f_{cx}$ in $\BBB_{100}(\tfrac{1}{3})$, obtained following Theorem 5.2 of \cite{fontana2024high}, and given by
			\begin{equation*}\label{pmfb}
				f_{cx}(\xx)=
				\begin{cases}
					\frac{1}{3},  &\xx=(\underbrace{1,\ldots, 1}_{33},0,\ldots, 0,0,\ldots,0) \\
					\frac{1}{3}, &\xx=(0,\ldots, 0,\underbrace{1,\ldots, 1}_{34},0,\ldots, 0) \\
					\frac{1}{3}, &\xx=(0,\ldots, 0,0,\ldots,0, \underbrace{1,\ldots, 1}_{33})
				\end{cases}.
			\end{equation*}
			The equicorrelation of the exchangeable $\Sigma_{cx}$-smallest element in $\GGG_{100}^{1/3}\left(Exp \left(\tfrac{1}{10}\right)\right)$ is $\rho_e =-0.0022$, that is the minimal correlation in the subclass of exchangeable distributions in $\GGG_{100}^{1/3} \left(Exp \left(\tfrac{1}{10}\right)\right)$.
			Let $A = \{1,\dots,33\}^2 \cup \{34,\dots,67\}^2 \cup \{68,\dots,100\}^2$. The entries of Pearson's correlation matrix of $\XX_{cx}$ are given by
			\begin{equation*}
				\rho_P(X_{j_1},X_{j_2})=
				\begin{cases}
					1, & j_1=j_2 \\
					\frac{2}{9}, & j_1 \neq j_2, (j_1,j_2) \in A \\
					-\frac{1}{9}, & j_1 \neq j_2, (j_1,j_2) \in \{1,\dots,100\}^2 \setminus A.
				\end{cases}.
			\end{equation*}
			The mean $\rho_{m}$ of Pearson's correlations of the random vector $\XX_{cx}$ is given by
			\begin{equation} \label{eq:rhom}
				\rho_m= \frac{2}{99 \times 100} \sum_{j_1 = 1}^{99} \sum_{j_2 = j_1 + 1}^{100} \rho_P(X_{j_1},X_{j_2})
				= -0.0022.
			\end{equation}
			From \eqref{eq:rhom} we notice that $\rho_m=\rho_e$, the equicorrelation of the exchangeable vector $\XX^e_{cx}$. This result is a consequence of Corollary~5.2 in \cite{fontana2024high} and of the fact that the Pearson's correlations in the class with the exponential margins is equal to the covariance of the corresponding Bernoulli pair.
			
		\end{example}

		\section{Remarks and conclusion}\label{NumILL}

		We conclude with one example in low dimension of the general class $\GGG_d(\pp)$, with $\pp=(p_1,\ldots, p_d)$, and we leave its theoretical investigation to further research.
		We find the risk measures' sharp bounds for the sum $S=X_1+X_2+X_3$, where $X_i$ have discrete distributions and $\XX$ has a GFGM copula with vector parameter $\pp$. 
		In fact, for $d=3$ we are able to find the extremal points of $\BBB_d(\pp)$, to construct the corresponding copulas, and to find the generators of the convex polytope $\GGG^{\pp}_d (F_1,\ldots, F_d)$. 
		We evaluate the risk measures on the extremal point and we find sharp bounds by enumeration. 
		Furthermore, we find the expected allocation and the expected contribution of $X_i$ for each risk  $\XX$ with extremal pmf, following \cite{blier2022generating}.

		\begin{example}\label{ex:final}
			
			We consider the class $\mathcal{B}_d(\pp)$ with $d=3$ and $\pp=(\tfrac{1}{2}, \tfrac{1}{3}, \tfrac{2}{3})$. 
			In this case there are $n_{\pp}^{\mathcal{B}} = 12$ extremal points that can be found by using \texttt{4ti2}, \cite{4ti2}.
			The extremal points $\boldsymbol{r}_k$, $k \in \{1,\ldots,12\}$, of the class $\mathcal{B}_3(\pp)$ are reported in Table~\ref{rays_ex1}.
			Note that there are three extremal pmfs whose sum is minimal under the convex order: $r_1$, $r_2$, and $r_4$. These three vectors have two couples of Bernoulli rvs with minimal covariance and the remaining pair --- $(I_1,I_2)$ for $r_1$, $(I_1,I_3)$ for $r_2$, and $(I_2,I_3)$ for $r_4$ --- has maximal covariance.
			
			\begin{table}[ht]
				\centering
				\begin{tabular}{c|cccccccccccc}
					\toprule
					$\xx$ & $\rr_{1}$ & $\rr_{2}$ & $\rr_{3}$ & $\rr_{4}$ & $\rr_{5}$ & $\rr_{6}$ & $\rr_{7}$ & $\rr_{8}$ & $\rr_{9}$ & $\rr_{10}$ & $\rr_{11}$ & $\rr_{12}$\\ 
					\midrule
					(0,0,0) & 0 & 0 & 0 & 0 & 0 & 0 & $\tfrac{1}{6}$ & $\tfrac{1}{6}$ & $\tfrac{1}{6}$ & $\tfrac{1}{3}$ & $\tfrac{1}{3}$ & $\tfrac{1}{4}$ \\ 
					(1,0,0) & 0 & 0 & $\tfrac{1}{6}$ & $\tfrac{1}{3}$ & $\tfrac{1}{3}$ & $\tfrac{1}{4}$ & 0 & $\tfrac{1}{6}$ & $\tfrac{1}{6}$ & 0 & 0 & 0 \\ 
					(0,1,0) & 0 & $\tfrac{1}{3}$ & 0 & 0 & 0 & $\tfrac{1}{12}$ & $\tfrac{1}{6}$ & 0 & 0 & 0 & 0 & 0 \\ 
					(1,1,0) & $\tfrac{1}{3}$ & 0 & $\tfrac{1}{6}$ & 0 & 0 & 0 & 0 & 0 & 0 & 0 & 0 & $\tfrac{1}{12}$ \\ 
					(0,0,1) & $\tfrac{1}{2}$ & $\tfrac{1}{6}$ & $\tfrac{1}{2}$ & $\tfrac{1}{6}$ & $\tfrac{1}{3}$ & $\tfrac{5}{12}$ & 0 & 0 & $\tfrac{1}{3}$ & 0 & $\tfrac{1}{6}$ & 0 \\ 
					(1,0,1) & $\tfrac{1}{6}$ & $\tfrac{1}{2}$ & 0 & $\tfrac{1}{6}$ & 0 & 0 & $\tfrac{1}{2}$ & $\tfrac{1}{3}$ & 0 & $\tfrac{1}{3}$ & $\tfrac{1}{6}$ & $\tfrac{5}{12}$ \\ 
					(0,1,1) & 0 & 0 & 0 & $\tfrac{1}{3}$ & $\tfrac{1}{6}$ & 0 & $\tfrac{1}{6}$ & $\tfrac{1}{3}$ & 0 & $\tfrac{1}{6}$ & 0 & $\tfrac{1}{4}$ \\ 
					(1,1,1) & 0 & 0 & $\tfrac{1}{6}$ & 0 & $\tfrac{1}{6}$ & $\tfrac{1}{4}$ & 0 & 0 & $\tfrac{1}{3}$ & $\tfrac{1}{6}$ & $\tfrac{1}{3}$ & 0 \\ 
					\bottomrule
				\end{tabular}
				\caption{\emph{Extremal points $\boldsymbol{r}_k$, $k \in \{1,\ldots,12\}$, of the class $\BBB_3(\tfrac{1}{2}, \tfrac{1}{3}, \tfrac{2}{3})$.} }
				\label{rays_ex1}
			\end{table}
			
			We now consider the class $\mathcal{G}_3^{\pp}(F_1, F_2, F_3)$, where $F_i$, $i \in \{1,2,3\}$, is the cdf whose pmf $f_i$ is defined by
			\begin{equation*}
				f_i(y)=
				\begin{cases}
					1-a_i , & y = 0 \\
					a_i [\big(\frac{y}{n}\big)^{c_i} - \big(\frac{y-1}{n}\big)^{c_i} ], & y \in \{1,\ldots,n\}
				\end{cases},               
			\end{equation*}
			and we choose $n=1000$, $a_1=0.2,\,\,\, a_2=0.1,\,\,\ a_3=0.3$ and $c_1=3,\,\,\, c_2=4,\,\,\, c_3=2$. Also, we have $E[X_1] = 150.09995$, $E[X_2] = 80.04997$, $E[X_3] = 200.14995$ and $E[S] = 430.29987$.

			The lower and upper bounds of convex risk measures are reached at the random vectors corresponding to the extremal points $\rr_1$ and $\rr_{11}$, respectively.
			In Table \ref{Table_corrX}, we provide Pearson's correlation coefficients $\rho(X_1,X_2)$, $\rho(X_1,X_3)$, and $\rho(X_2,X_3)$ for all of the twelve extremal dependence structures. Notice that the correlation matrix of $\XX_1$ also has negative entries, i.e $(X_{1,1}, X_{1,3})$ and $(X_{1,2}, X_{1,3})$ are negatively correlated, while  $(X_{1,1}, X_{1,2})$ has the same correlation as $(X_{11,1}, X_{11,2})$. This last equality follows observing that the correlation $\rho_1(1,2)$ between $(r_{1,1}, r_{1,2})$ and the correlation $\rho_{11}(1,2)$ between $(r_{11,1}, r_{11,2})$ are equal, in fact
			
			\begin{equation*}
				\rho_1(1,2)=\frac{r_1((1,1,0))+r_1((1,1,1))-\frac{1}{2}\frac{1}{3}}{\sqrt{\frac{1}{2}(1-\frac{1}{2})\frac{1}{3}(1-\frac{1}{3})}} = 0.0605,
			\end{equation*}
			and 
			\begin{equation*}
				\rho_{11}(1,2)=\frac{r_{11}((1,1,0))+r_{11}((1,1,1))-\frac{1}{2}\frac{1}{3}}{\sqrt{\frac{1}{2}(1-\frac{1}{2})\frac{1}{3}(1-\frac{1}{3})}} = 0.0605,
			\end{equation*}
			coincide since $r_1((1,1,0))+r_1((1,1,1))=r_{11}((1,1,0))+r_{11}((1,1,1))=\frac{1}{3}$.

			\begin{table}[ht!]
				\centering
				\begin{tabular}{c|cccccc}
					\toprule
					& ${\rr_1}$ & ${\rr_2}$ & ${\rr_3}$ & ${\rr_4}$ & ${\rr_5}$ & ${\rr_6}$  \\ 
					\midrule
					$\rho(X_1,X_2)$ & 0.0605 & -0.0605 & 0.0605 & -0.0605 & 0.0000 & 0.0302 \\
					$\rho(X_1,X_3)$ & -0.1610 & 0.1610 & -0.1610 & -0.1610 & -0.1610 & -0.0805 \\ 
					$\rho(X_2,X_3)$ & -0.1229 & -0.1229 & -0.0307 & 0.0615 & 0.0615 & 0.0154 \\ 
					\bottomrule
					\toprule
					& ${\rr_7}$ & ${\rr_8}$ & ${\rr_9}$& ${\rr_{10}}$ & ${\rr_{11}}$ & ${\rr_{12}}$  \\
					\midrule
					$\rho(X_1,X_2)$ & -0.0605 & -0.0605 & 0.0605 & 0.0000 & 0.0605 & -0.0302 \\
					$\rho(X_1,X_3)$ & 0.1610 & 0.0000 & 0.0000 & 0.1610 & 0.1610 & 0.0805 \\ 
					$\rho(X_2,X_3)$ & -0.0307 & 0.0615 & 0.0615 & 0.0615 & 0.0615 & 0.0154 \\ 
					\bottomrule
				\end{tabular}
				\caption{Pearson's coefficients of $\XX$.} \label{Table_corrX}
			\end{table}
			
			\begin{table}[ht]
				\centering
				\begin{tabular}{r|rrrrrr}
					\toprule
					& ${\rr_1}$ & ${\rr_2}$ & ${\rr_3}$ & ${\rr_4}$ & ${\rr_5}$ & ${\rr_6}$  \\ 
					\midrule
					$VaR_{0.95}(S)$ &\squared{ 1219.00} & 1532.00 & 1360.00 & 1342.00 & 1403.00 & 1479.00 \\ 
					$ES_{0.95}(S)$ & \squared{1590.08} & 1733.70 & 1665.46 & 1641.07 & 1683.14 & 1724.32 \\ 
					$\Psi_{0.001}(S)$ & \squared{555.98} & 587.74 & 566.80 & 563.46 & 570.51 & 580.07 \\
					$Std(S)$ & \squared{473.23} & 521.70 & 488.85 & 485.22 & 494.70 & 508.47 \\ 
					\bottomrule
					\toprule
					& ${\rr_7}$ & ${\rr_8}$ & ${\rr_9}$& ${\rr_{10}}$ & ${\rr_{11}}$ & ${\rr_{12}}$  \\ 
					\midrule
					$VaR_{0.95}(S)$ & 1561.00 & 1493.00 & 1567.00 & 1618.00 & \underline{1643.00} & 1535.00 \\ 
					$ES_{0.95}(S)$ & 1802.17 & 1771.05 & 1824.07 & 1888.55 & \underline{1906.84} & 1818.89 \\ 
					$\Psi_{0.001}(S)$ & 602.12 & 590.22 & 603.90 & 622.97 & \underline{629.61} & 601.55 \\ 
					$Std(S)$ & 535.91 & 518.49 & 536.10 & 558.13 & \underline{566.39} & 531.66 \\ 
					\bottomrule
				\end{tabular}
				\caption{\emph{Values of risk measures of the sum of the components of vectors with joint cdf defined by the GFGM copulas corresponding to the extremal probability mass functions ($\rr_i$, $i \in \{1,\ldots, 12\}$) of the class $\mathcal{B}_3(\tfrac{1}{2}, \tfrac{1}{3}, \tfrac{2}{3})$. The minimum values are squared and the maximum ones are underlined.}}
			\end{table}
			We conclude this example by finding the contribution of risk $X_j$, $j \in \{1,2,3\}$, to the standard deviation of the sum $S=X_1+X_2+X_3$, to the $\VaR_{\alpha}$ and to the $\ES_{\alpha}$, for all the extremal dependence structures.
			We recall the definitions of expected allocation and of expected contribution of the risk $X_j$ to a total outcome $S=y$, for $y \in \{0,1, \ldots, 3000\}$. The  expected allocation of each risk $X_j$ in Definition 1.1 of \cite{blier2022generating} is given by
			\begin{equation*}
				E[X_j\boldsymbol{1}\{S=y\}],
				\quad
				y \in \mathbb{N}_+,
			\end{equation*}
			where $\boldsymbol{1}$ is the indicator function, 
			such that $\boldsymbol{1}\{A\} = 1 $, if $A$ is true, 
			and  $\boldsymbol{1}\{A\} = 0 $, otherwise.
			The expected contribution of each risk $X_j$, $j \in \{1,2,3\}$, is provided in Equation (6) of \cite{blier2022generating}, and is defined by
			\begin{equation*}
				E[X_j|S=y] = \frac{E[X_j\boldsymbol{1}\{S=y\}]}{\Pr(S=y)},
				\quad
				y \in \mathbb{N}_+,
			\end{equation*}
			assuming that $\Pr(S = y) > 0$.
			The expected contribution of $X_j$ to the $\VaR_{\alpha}$ is given by $E[X_j|S=\VaR_{\alpha}(S)]$.
			We now recall the expression for the contribution to the $\ES_{\alpha}$ based on the Euler-based allocation rule provided in \cite{tasche1999risk}:
			\begin{equation*}
				CES_{\alpha}(X_j, S)=\frac{E[X_j] - E[X_j \boldsymbol{1}\{S\leq \VaR_{\alpha}(S)\}] + \beta_S E[X_j \boldsymbol{1}\{S = \VaR_{\alpha}(S)\}]}{1-\alpha},
			\end{equation*}
			where
			\begin{equation*}
				\beta_S 
				= 
				\begin{cases}
					\frac{\Pr(S\leq  \VaR_{\alpha}(S)) - \alpha}{\Pr(S = \VaR_{\alpha}(S))}, & \text{if } \Pr(S = \VaR_{\alpha}(S)) > 0, \\
					0, & \text{if } \Pr(S = \VaR_{\alpha}(S)) = 0,  
				\end{cases}
			\end{equation*}
			and 
			\begin{equation*}
				E[X_j \boldsymbol{1}\{S\leq k\}]
				=
				\sum_{y=1}^k E[X_j \boldsymbol{1}\{S = y\}], 
				\quad 
				k \in \mathbb{N}_+.
			\end{equation*}
			Finally, the contribution of $X_j$ to the standard deviation of $S$ based on Euler's rule is given by
			\begin{equation}
				CStd(X_j,S) 
				=
				\frac{Cov(X_j,S)}{\sqrt{Var(S)}}
				=
				\frac{Var(X_j) + \sum_{j' \neq j} Cov(X_j, X_{j'}) }{\sqrt{Var(S)}},
				\quad
				j = 1,2,3.
			\end{equation}
			Figure \ref{barplot_Contributions} reports the contributions to the $\VaR_{0.95}$, the $\ES_{0.95}$ and the standard deviation of $S$.
			
			\begin{figure}[tb]
				\centering
				\includegraphics[width=0.8\linewidth]{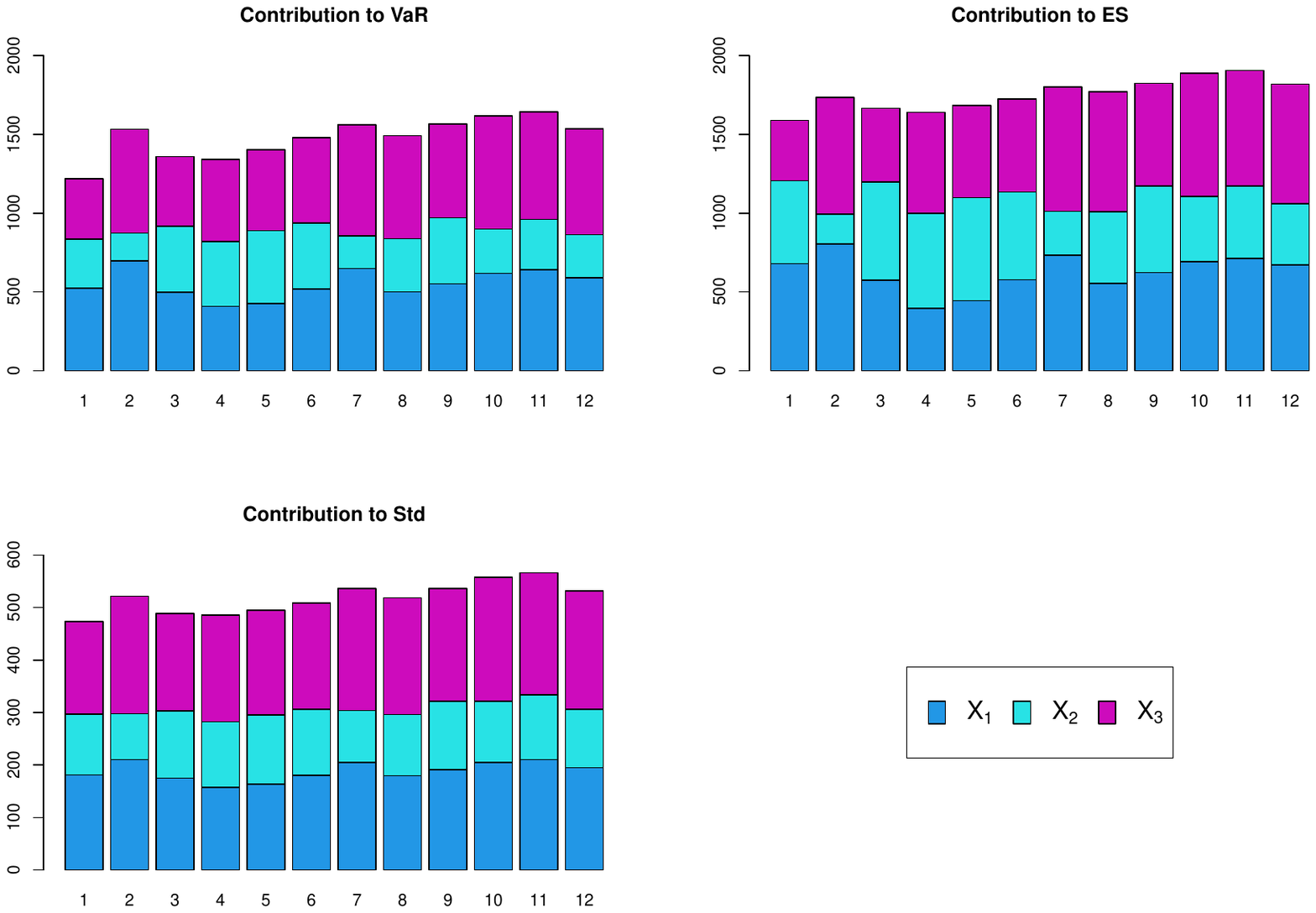}
				\caption{\emph{Contributions of $X_j$, $j \in \{1,2,3\}$ to the $\VaR_{0.95}(S)$ (top-left), to the $\ES_{0.95}(S)$ (top-right), and to the standard deviation of $S$ (bottom-left), based on Euler's rule.}}
				\label{barplot_Contributions}
			\end{figure}

		\end{example}
		
		The geometrical structure of GFGM copulas inherited from the geometrical structure of multivariate Bernoulli distributions has proven to be a powerful tool for studying the properties of random vectors with GFGM dependence. 
		
		The last Example \ref{ex:final} finds the bounds by enumeration of their values in the extremal points, which becomes computationally challenging in high dimensions.  Under the assumption of identically distributed risks with GFGM$(p)$ dependence structure, we show the effectiveness of our theoretical results in studying the risk of high dimensional --- $d=100$ --- portfolios.   
		The extension of these theoretical results to the whole GFGM copulas relies on extending corresponding results in the class of multivariate Bernoulli distributions, and this is part of our ongoing research. Another, more applicative, part is to use this novel geometrical representation to investigate the dependence structure of the class and of their extremal points, which are good candidates for representing extremal dependence also in high dimension.
		
		\newpage
		
		\newpage
		\section{Acknowledgements}
		
		This work was partially supported by the Natural Sciences and Engineering Research Council of Canada (Cossette: 04273; Marceau: 05605). This work was also partially supported by the Italian Ministry of Education, University and Research (MIUR), PRIN 2022-PNRR project P2022XT8C8. H. Cossette and E. Marceau would like to thank \textit{Dipartimento di Scienze Matematiche "G. L. Lagrange" (DISMA), Politecnico di Torino,} for their wonderful stay during which most of the paper was written.

		\bibliographystyle{apalike}
		\bibliography{ref}

	\end{document}